\documentclass[11pt]{article}
\usepackage{geometry}                
\usepackage{graphicx}
\usepackage{amssymb}
\usepackage{graphicx}
\usepackage{natbib}
\usepackage{comment}
\usepackage{authblk}

 \usepackage{mathptmx}      

\usepackage{amsmath}
\usepackage{amsthm}
\usepackage{bm}
\usepackage{bbold}
\usepackage{hyperref}

\newtheorem{theorem}{Theorem}
\newtheorem{lemma}{Lemma}
\theoremstyle{definition}

\newtheorem{example}{Example}[section]

\newtheorem{remark}{Remark}[section]

\DeclareRobustCommand{\stir}{\genfrac\{\}{0pt}{}}

\def\a{\alpha}

\def\b{\beta}
\def\la{\lambda}
\def\th{\theta}
\def\nn{\notag}
\def\ve{\varepsilon}

\def\btp{\tilde{\bm p}}
\def\btq{\tilde{\bm q}}
\def\be{\bm \varepsilon}

\def\TS{\mathcal{T}}
\def\D{\mathcal {D}_\a}
\def\H{\mathcal {H}_\a}
\def\C{\mathcal{CV}}

\def\p{\hat{p}}
\def\q{\hat{q}}
\def\bp{{\bm p}}
\def\bq{{\bm q}}
\def\bhp{\hat{\bm p}}
\def\bhq{\hat{\bq}}
\def\hbp{\bhp}
\def\hbq{\bhq}
\def\bu{{\bm u}}
\def\u{\hat{u}}
\def\bhu{\hat{\bm u}}
\def\S{{\cal S}_\a}
\def\W{\tilde{W}}
\def\G{{\cal G}}
\def\ch{{\cal X}^2}
\def\R{\tilde{R}}

\def\chu{{\cal X}_{\bu}^2}
\def\ch{{\cal X}^2}
\def\chp{\ch_\bp}

\newcommand\tabT{\rule{0pt}{2.6ex}}
\newcommand\tabB{\rule[-1.2ex]{0pt}{0pt}}
\newenvironment{cenum}{\ \newline \begin{list} 
{\it {(\roman{enumi})}} 
{\topsep=0in\itemsep=.05in\parsep=0in\leftmargin=10pt\itemindent=5pt\usecounter{enumi}}}{\end{list}}

\numberwithin{equation}{section}

\begin{document}

\title{Limit Theorems  for Empirical R\'enyi Entropy and  Divergence  with Applications to Molecular Diversity Analysis\thanks{Research partially supported by US NIH grants  R01CA-152158 and U01-GM092655  (GAR, MS, MP)  as well as  US NSF grant DMS-1106485 (GAR).}
}


\author[1]{Maciej Pietrzak (pietrzak.20@osu.edu)}
\author[1,2]{Grzegorz A. Rempala (rempala.3@osu.edu)}
\author[3]{Micha{\l} Seweryn (msewery@math.uni.lodz.pl)}
\author[4]{Jacek Weso{\l}owski (wesolo@mini.pw.edu.pl)}
\affil[1]{Division of Biostatistics, College of Public Health,  The Ohio State University, Columbus OH, USA }   
\affil[2]{Mathematical Biosciences Institute, The Ohio State University, Columbus OH, USA}
\affil[3]{Department of Mathematics and Computer Science, University of {\L}\'odz, Poland}
\affil[4]{Wydzia{\l} Matematyki i Nauk Informacyjnych, Politechnika Warszawska, Warsaw, Poland}

\date{\today}
\maketitle

\begin{abstract} 
Quantitative methods for studying biodiversity have been traditionally rooted in the classical theory of finite frequency tables analysis. However, with the help of modern experimental tools, like high throughput sequencing, we now begin to unlock the outstanding diversity of genomic data in plants and animals reflective of the long evolutionary history of our planet. This molecular data often defies the classical frequency/contingency tables assumptions and seems to require sparse tables with very large number of categories and highly unbalanced cell counts, e.g., following heavy tailed distributions (for instance, power laws). Motivated by the molecular diversity studies, we propose here a frequency-based framework for biodiversity analysis in the asymptotic regime where the number of categories grows with sample size (an infinite contingency table). Our approach is rooted in information theory and based on the Gaussian limit results for the effective number of species (the Hill numbers) and the empirical Renyi entropy and divergence. We argue that when applied to molecular biodiversity analysis our methods can properly account for the complicated data frequency patterns on one hand and the practical sample size limitations on the other. We illustrate this principle with two specific RNA sequencing examples: a comparative study of T-cell receptor populations and a validation of some preselected molecular hepatocellular carcinoma (HCC) markers.
\par {\bf keywords} Hill number,  Central limit theorem,  Next generation sequencing, Triangular arrays, T-cell receptors 
\par {\bf AMS classification} {60F05  60G42  94A17}
\end{abstract}

\section{Introduction}

Developing effective methods for quantifying and comparing empirical diversity of various biological populations is one of the fundamental problems  of modern life sciences as it has  direct impact on our understanding of the basic  operating principles of our planet's  ecosystem and its evolution \citep[cf., eg.,][]{Berkov:2014fk}.
In the  course  of its 3.5 billion years of evolutionary history, nature has developed an outstanding bio- and molecular diversity among  the Earth's  species of plants and animals.   Indeed,  it is  estimated that there are currently about 8.7 million eukaryotic species on earth,  both  marine and 
terrestrial,    88\% of which are still waiting to be described \citep{Mora:2011xy}.  The diversity  at the molecular level is 
perhaps even more spectacular,   as it occurs at different levels of biological organization: 
within one individual (e.g., through RNA, DNA, proteins, and metabolites), between individuals of the same and related species,  within and between species  and ecosystems, as well as  throughout evolution \cite[see, e.g.,][]{campbell2003save}.
For instance,   the  number of different molecular types of human T-cells is estimated at $10^{18}$ \citep{Jane05} which  only slightly less that the currently estimated number of stellar objects in the known universe (the latter believed to be of  the order $10^{21}$).  
 
Whereas the power of modern  computing has  allowed us  to make  steady progress towards  building  ever more  robust empirical measures of biodiversity  based on  a variety of  considerations \cite[see, e.g.,][]{Presley:2014fr},  the most relevant to our discussion here are the  measures  borrowed from  the field of information theory.   They  include among others  the   {\em Hill number}  (or the effective number of species) and the related concept of  the {\em Renyi  entropy}   (see, e.g., the recent review \cite{chiu2014phylogenetic} and references therein). Although  originally proposed  for  quantifying ecological diversity in the macro-scale ecosystems \citep{Chao:2010rt},   the use  of the empirical Renyi entropy as a descriptor of diversity  was also adopted  for  molecular populations in \citet{DeAndrade11}.     Since then   the Renyi-type  measures were   applied to  problems of molecular populations ranging from  analyzing   regulatory variants   and   testing   genome-wide associations \citep{Sun13,Sadee:2014aa} to comparing  different   T-cell populations  \citep{Cebula2013,RS12}.  Despite their growing usage in biodiversity studies of  both macro- and molecular- level populations,  it appears that  some important   statistical properties  of  the Renyi-type  measures have not been yet sufficiently understood, especially in the context of frequency-based   analysis and large sample behavior.

Currently,   standard  methods of obtaining  molecular level data on  the {\em transcriptome} (RNA)  abundance  rely on the so-called  next-generation sequencing (NGS) technology and especially    the high-throughput  RNA sequencing or RNA-seq \citep{Wang:2009aa}. However, the   molecular count  data  from NGS often elude  standard  statistical  analysis  due to the fact that exhaustive sampling of the DNA and RNA fragments for the purpose of sequence reconstruction  is  not feasible and that the  sequencing errors increase  with  sampling intensity or  {\em sequencing depth} \citep{ORawe:2015kq}.    It has been therefore generally conceded  \citep{Oh:2014aa}  that  the  standard,   fixed-dimension,  non-parametric frequency/contingency table  analysis   (see, e.g., \citealt{AA02}) does not readily apply  to the NGS data and that a different, {\em infinite-size}  contingency table framework,   more  reflective of the current sequencing technology,  appears  necessary.  Due to the  nature of the  NGS methods,   such  framework  should be based on the large sample (high-throughput)  considerations  but,  at the same time,  should also  account   for  the  increase in   the number of sequencing errors   with  increasing sample size  as well as   for the  under-sampling bias.   

 Motivated by  the questions on comparing biodiversity in molecular data (especially arriving from  the NGS experiments)  in  the current paper we establish   some  large sample results for the empirical Renyi entropy and divergence  in order to  bridge the gap between  current heuristic approaches   and  a more formal statistical  theory of large samples.   To this end, we  derive herein several central  limit theorems (CLTs)  which  yield approximate confidence bounds  for the  (Renyi) entropy-based  measures of diversity  and similarity in  the setting of  an infinite contingency table. Our CLT results   complement both the  law of large number theorems in \cite{RS12}  as well as the  CLT   for the plugin estimates of the Shannon entropy \cite{zhang2012normal} and the Kullback-Leibler divergence estimates  \citep{paninski2003estimation,zhang2014nonparametric}.  Since  in the  NGS experiments   one typically expects to under-sample the transcriptome, we focus here on   the  Renyi entropy exponent  (which below is denoted by $\alpha$) less than one, so as to   up-weight the  contributions of the lower  counts and  our  CLT results are restricted to this case. The extensions to  arbitrary exponents are straightforward but  not considered here.  In order to provide examples of the types of  applications  motivating   the mathematical results, we  analyze   two real  biological datasets from two different types of  NGS experiments.  In the first experiment, described  in   the  study \cite{Cebula2013},   one compares   multiple  T-cell receptors populations taken from mice before and after treatment with antibiotics.  The goal of  the  second  experiment  is the  elucidation  of   differences in  gene expression profiles between cancer and  control tissues  in individuals with hepatocellular carcinoma,  as described  in \cite{HCC14}. In both presented examples  the NGS datasets are analyzed  and de-noised by applying   a multi-stage process  developed on the basis of our theoretical results. 
 
As already indicated above, the problem of empirically estimating entropy and divergence  has been extensively studied in the statistical and machine learning literature  over past several decades, both  in the context of  discrete  and continuous   distributions.  See, for instance,  the monograph by \cite{pardo2005statistical} or the review  in \cite{Krish14} for more details.  In the  general case of   Renyi's  entropy and closely related Tsallis' entropy  of  a fixed continuos  distribution $f$ in ${\mathbb R}^m$,   a class of consistent estimators   was  proposed in \cite{leonenko2008class}  based on the $k$-th nearest-neighbor distances computed from the appropriate random samples of size $n$ from $f$.  The idea was later also extended to the Renyi entropy functionals in \cite{kallberg2012statistical} and it appears that similar results could be expected  to hold in the discrete case as well.  The main difference between these types of results and what is considered here is that in our setting  the discrete density function $f$ is allowed to change as the  sample size $n$ increases. Additionally, although   in the current and that we only analyze the basic  empirical frequency (the so-called plug-in) estimates.

The paper is organized as follows. In the next section  (Section~2) we outline    the relevant mathematical   concepts  along with the necessary notation. In  Section~3 we state the main theoretical results of the paper, namely the CLTs for the Hill number (or the Tsallis entropy) and  the Renyi entropy and divergence  in the  asymptotic regime when the diversity of the population (i.e., the number of different types) grows with the sample size.  The results for the simpler case (Theorems 1 and 2) when  Renyi entropy statistics admit linear approximations are  established via the intermediate CLT results for the corresponding power sums which are closely related to the CLTs for  Hill's numbers and Tsallis' entropies. These results are also  included  as parts of formulations of Theorems 1 and 2. In case of  the uniform distribution for  the  Renyi entropy  as well as   the equal-marginals bivariate distribution for  the Renyi divergence, the power sum CLTs are no longer valid (there is no linear approximation available) and other methods are required to establish weak convergence to Gaussian variates  under slightly more stringent conditions. These results  are  presented  as Theorems 3 and 4 in Section~3. As it turns out, the key ingredient needed to establish Theorems 3 and 4 is the CLT result for two  Pearson-type  chi-square statistics in an  infinite contingency table. This latter result is of interest in itself and is presented as  Lemma~2 in Section~3.   In the following Section~4, we provide some simulation-based  examples  of the asymptotic behavior  of estimates from Section~3 in the case (relevant for our applications)  of  power law distributions under   various  sampling scenarios.  These examples illustrate  in particular how the  CLTs of Section~3  may  hold or not, depending on the relations between the dimensions of the relevant contingency tables and the empirical sample sizes.  In the second part  of Section~4  we  also discuss in detail the two biological examples of NGS data analysis and  show how the results of Section~3 may be used to analyze biodiversity of T-cell receptors  and to  profile the multiple sets of transcriptomes.  The final Section~5 offers a summary and brief conclusions. The proofs of all more complicated results are provided in the appendix along with some auxiliary technical lemmas.

\section{Power Sums,  Entropy and Divergence}\label{sec:2}
Consider a triangular array of bivariate row-wise independent  random variables  $Z_{n,k}$ for $k=1,\ldots,n$ which in each row are equidistributed with the random variable  $Z_n=(X_n,Y_n)$  such that $P(X_n=i,Y_n=j)=p_{ij}^{(n)}$ for $i,j=1,\ldots,m_n$. Below we suppress the index $n$  when possible,  writing e.g., $m, Z_k, Z, p_{ij}$, etc.  for simplicity.  

Let  $\a>0$ and for any  probability distribution  $\bp=(p_i)_{i=1}^m$ 
 define
 \begin{equation}\label{eq:1} 
 \S(\bp) =\sum_{i=1}^m {p}_i^\alpha.
 \end{equation}
 Similarly, for any pair  of distributions $\bp=(p_i)_{i=1}^m$ and $\bq=(q_i)_{i=1}^m$
 define 
  \begin{equation}\label{eq:2} \S(\bp,\bq) =\sum_{i=1}^m {p}_i^\alpha q_i^{1-\a}.
  \end{equation} 
   (Note that  ${\cal S}_1\equiv 1$). 
   The well-known special case of the above is  $\a=1/2$,  which results in  a symmetric index  ${\cal S}_{1/2}(\bp,\bq)={\cal S}_{1/2}(\bq,\bp)$ often referred to as  the  Bhattacharyya coefficient  \citep[see, e.g.,][]{nielsen2011burbea}.
   
    Recall   \citep{Renyi63} that for a given distribution $\bp$ its  Renyi entropy $\H$  is defined as
   $$\H(\bp) =\frac{1}{1-\a} \log\left( \sum p_i^\a\right)=\frac{1}{1-\a} \log \S(\bp) $$ and that for  a pair 
   of  distributions $(\bp,\bq)$  their   Renyi divergence $\D$ is  defined  as 
   $$\D(\bp,\bq) = \frac{1}{\a-1}\log \S(\bp,\bq).$$  Note that the sign change in the normalizing constant is needed in order to ensure non-negativity of $\H$ and $\D$. The special case of $\D$ with $\a=1/2$ is referred to as the Bhattacharyya distance,  and may be expressed in terms of the   Mahalanobis distance \citep[see, e.g.,][]{nielsen2011burbea}, whereas      the linear approximation of $\H(\bp)$  given by 
\begin{equation}\label{eq:tse}\TS(\bp)=\frac{1}{1-\a}(\S(\bp)-1).\end{equation} 
 is sometimes referred to as the  Tsallis entropy  and has important applications in the field of statistical mechanics  \citep{tsallis1988possible}.   
 Note that for our current purposes,  we  will only consider the  quantities $\D,\H$, and $\TS$  for  $\a$ satisfying $0<\a<1$. 
 
 In what follows the summation symbol  without subscripts ($\sum $) will indicate  summation with respect to the index $i$   ($i=1,\ldots,m$)  whereas $\bp=(p_i)_{i=1}^m$ and $\bq=(q_i)_{i=1}^m$ will  (typically) denote the marginal distributions  of  the bivariate variable $Z=(X,Y)$ whose distribution is denoted by $(p_{ij})_{i,j=1}^m$.  Additionally, the  uniform distribution on $m$ points will be denoted by $\bu$.  An important   relation  between  the Renyi entropy and the Renyi divergence   is 
 \begin{equation}\label{eq:hud}
 \H(\bp)=\log m -\D(\bp,\bu).
 \end{equation} 
 We  note   also the following monotonicity property of $\D$ and $\H$ with respect to the index $\a$. 
 \begin{lemma} For $0<\a< \b<1$ we have 
 $\D(\bp,\bq)\le {\cal D}_\beta(\bp,\bq) $ and thus, in view of \eqref{eq:hud}, also $ \H(\bp)\ge {\cal H}_{\beta}(\bp)$.
 \end{lemma} 
 \begin{proof} Note that  for $x \ge 0$ the function $x \rightarrow x^\frac{\a-1}{\beta-1}$ is strictly convex for $0<\a<\b<1$. Therefore, by Jensen's inequality  \begin{align*}\D(\bp,\bq) &=\frac{1}{\a-1} \log \sum p_i^\a q_i^{1-\a} =\frac{1}{\a-1} \log \sum p_i\left(\frac{q_i}{p_i}\right)^{(1-\beta)\frac{\a-1}{\beta-1}}\\ &\le  \frac{1}{\beta-1} \log \sum p_i\left(\frac{q_i}{p_i}\right)^{(1-\beta)}={\cal D}_\beta(\bp,\bq).\end{align*} 
 \end{proof}
\begin{example}[Hill's Number] For given $0<\a<1$  the measure of diversity of a   distribution  $\bp$  also known as the  {\em effective number of classes } may be defined as  \citep[see, e.g.,][]{jost2007partitioning,Chao:2012mz,RS12}  $ENC_\a(\bp)=\exp (\H(\bp))=\S(\bp)^{1/(1-\a)}$. It follows then from Lemma~1 that for any $0<\a< \b<1$ we have 
 $ENC_\a(\bp) \ge ENC_\b(\bp)$. (As it turns out, this inequality may be in fact extended to arbitrary positive $\a<\b$). \end{example}

\subsection{Low Diversity  Condition and Projection Variables}\label{sec:ld}
The notion of an infinite-dimension contingency table  brought up in the introduction may be now  formally introduced simply  as a requirement that for $n$-size sample from $(p_{ij})_{i,j=1}^m$ we have $m\to \infty$ as $n\to \infty$. Throughout the paper, let $a\wedge b$ denote $\min(a,b)$ for any real $a,b$ and let $a_n\sim b_n$ (resp. $a_n\sim O(b_n)$) denote $a_n/b_n\to 1$ (resp. $A<\limsup_n a_n/b_n<B$ for some finite  $A,B$) as $n\to \infty$   for any  real sequences $a_n, b_n$. 
Throughout the paper  we  consider  only  the  {\em  low diversity} (LD) schemes in which the  marginals $\bp,\bq$,  of $Z$   satisfy  the following {\em LD condition}. \begin{equation}\label{eq:ld}   (n p_\ast)^{-1}=o(n^{-\tau})  \quad \text{for some}\quad \tau>0,  \end{equation} where $p_\ast=\min_i(p_i)\wedge\min_i(q_i)$.
Note that  since $p_\ast\le 1/m$  then  \eqref{eq:ld}  implies  in particular $m/n= o(n^{-\tau})$. As it turns out,  for many distributions  $\bp$ the two conditions are in fact  equivalent, as seen  in the  following. 

\begin{example}[Power Law Model] Let $\bp=\bq$ and 
assume that $p_i=H^{-1}(\beta,m)/(i^{\beta}l(i))$, $(i=1,\ldots,m)$ where $\beta>0 $, $l(x)$ is a non-decreasing slowly varying function (see, e.g., \citealt{soulier2009some}, chapter 1), and $H^{-1}(\beta,m)=1/\sum_{i=1}^m ({i^{\beta}l(i)})^{-1}$ is the normalizing constant. Note that if  $0<\beta<1$ then   $H^{-1}(\beta,m) \sim (1-\beta) l(m)/m^{1-\beta}$ and \eqref{eq:ld} 
is implied by  $m/n=o(n^{-\tau})$ since 
$$(n \min_i p_i)^{-1} \sim (1-\beta)^{-1}\frac{ m^\beta l(m)}{n m^{\beta-1} l(m)} = (1-\beta)^{-1}\frac{m}{n}.$$
\end{example}

For any $0<\alpha<1$ and a given pair $(m, n)$, let  us define  two random variables which will play an important role in the following section.  Let  $W_n^{(\a)}$ be   defined  as  
\begin{equation}\label{eq:w} P(W_n^{(\a)}=\a p_{i}^{\a-1})=p_{i}\end{equation} for $i=1,\ldots,m$. Similarly,  define also  $V_n^{(\a)}$  as
\begin{equation}\label{eq:v}
P\left(V_n^{(\a)}=\a \left(\frac{q_i}{p_i}\right)^{1-\a} + (1-\a) \left(\frac{p_j}{q_j}\right)^{\a}\right)=p_{ij} \end{equation} for $i,j=1,\ldots,m$. 
In the following,  for the reasons discussed below, we refer to \eqref{eq:w} and \eqref{eq:v} as the {\em projection variables} or simply {\em projections}. 
\begin{remark}\label{rem1}   Note that $$E W_n^{(\a)}=\a\S(\bp) $$ and $Var W_n^{(\a)}=0$ iff $p_i=1/m$ for all $i$, that is,  $\bp=(p_i)=\bu$ is a uniform distribution   on  $m$ support points (this case  is often referred to as  a maximal diversity model or a pure noise model).  Similarly, $$E V_n^{(\a)} = \S(\bp,\bq)$$ and it is also easy to see that 
$Var V_n^{(\a)}=0$ iff $p_{i}=q_i$ for all $i$, that is, $\bp=\bq$.
\end{remark}
As it turns out, both   cases $\bp=\bu$ and $\bp=\bq$  require special consideration in  the asymptotic analysis of $\H$ and $\D$. In view of the remark above they may be referred to as the cases of ``degenerate" (zero variance)  projections.  
\begin{example}[Noise--and--Signal and Pure Noise Models]\label{ns1}
 A distribution concentrated on $m+1$ support points, such that $p_0>0$ and $p_i=(1-p_0)/m$ for $1\le i\le m$, may be considered as a simple model of  signal contamination. Note that in this case we have 
$P(W_n^{(\a)}=\a p_0^{\a-1})=p_0$, $P(W_n^{(\a)}=\a m^{1-\a}(1-p_0)^{\a-1})=1-p_0$ and 
$$ Var W_n^{(\a)} = \a^2\left( m^{1-\a}(1-p_0)^\a\left(\frac{p_0}{1-p_0} \right)^{1/2}-p_0^\a\left(\frac{1-p_0}{p_0} \right)^{1/2}\right)^2.$$ 
For the pure noise model $p_0=0$,  in which case the support reduces to $m$ points, and the above formula is not valid. However,  as already pointed out  before, in this case  we may show directly that $Var W_n^{(\a)}=0$. \end{example}
 \section{Limit Theorems}\label{sec:lims} Let $N(0,1)$ denote the standard Gaussian random variable and $\Rightarrow$ denote the usual weak convergence in the  space of  probability  distributions.   Define also the plug-in $n$-sample estimates of $\bp$ and $\bq$ as, respectively,  
$\hbp=(\hat{p}_i)_{i=1}^m,$ where $ \hat{p}_i=\sum_{k=1}^n I(X_k=i)/n$ and  $\hbq=(\hat{q}_i)_{i=1}^m,$ where $ \hat{q}_i=\sum_{k=1}^n I(Y_k=i)/n$. Here and elsewhere in the paper $I(\cdot)$ denotes the indicator function.
As it turns out,  
two distinct sets of CLTs may be derived depending on whether the variables $W_n^{(\a)}$ and $V_n^{(\a)}$ are degenerate  (that is, their respective variances vanish) or not.  For the non-degenerate case  the appropriate CLTs may be established  by expanding on  the   usual projection and Taylor's expansion arguments \citep[see, e.g.,][chapter 1]{Shao04}. This is the simpler case to consider and we discuss it first. 
\subsection{CLTs for Non-Degenerate Projections}
The first two CLT results  for the empirical (plug-in) Renyi entropy and divergence  and their  corresponding power sums  are provided in Theorems~1 and 2 below. Their respective hypotheses  $(iii)$  may be viewed  as complementing  the analogous results established for the Shannon entropy and  the Kullback-Leibler divergence \citep{paninski2003estimation,zhang2012normal,zhang2014nonparametric}.
 Note also that   
$\S=(ENC_\a)^{1-\a}$ where the Hill number $ENC_\a$ is defined in Example~1.  The proofs are deferred  to the appendix. 

Recall that for any  square integrable random variable $X$, such that $EX\ne 0$,  we define its coefficient of variation as ${\cal CV}(X)=(Var X)^{1/2}|E X|^{-1}$.  \begin{theorem}[Renyi Entropy CLT]\label{thm:1}  Let $W_n^{(\a)}$ be a sequence of random variables defined by \eqref{eq:w} with  $\inf_n \C(W_n^{(\a)})>0$ and   let \begin{equation}\label{eq:pc} \sum p_i^{\a-1}(n Var W_n^{(\a)})^{-1/2}\to 0\quad\text{for}\   m,n\to \infty.\end{equation}
Then, under the LD condition \eqref{eq:ld}, as $m,n\to \infty$
\begin{itemize}
\item[(i)] $\S(\bhp)/\S(\bp) \to 1$ in probability, 
\item[(ii)] 
$\sqrt{n}(\S(\bhp)-\S(\bp))/(Var W_n^{(\a)})^{1/2}\Rightarrow N(0,1),$
\item [(iii)] $\sqrt{n}\,(1/\a-1)(\H(\bhp)-\H(\bp))/\C (W_n^{(\a)})\Rightarrow N(0,1).$
\end{itemize}
\end{theorem}
\begin{remark}Note that the first two assertions of the theorem may be equivalently stated in terms of the convergence of the Tsallis plug-in entropy defined by \eqref{eq:tse}.
\end{remark}
\begin{remark}\label{rem2}
Note that the  condition  \eqref{eq:pc} is typically stronger than \eqref{eq:ld}. Indeed, taking $\alpha>1/2$ and  the power law model from Example~2 with $0<\beta<1$ we obtain
 $\sum p_i^\a\sim (1-\beta)^\a\, m^{1-\a}/(1-\a\beta)$ and $\sum p_i^{2\a-1}\sim (1-\beta)^{2\a-1}\, m^{2-2\a}/(1-2\a\beta+\beta)$. Consequently, for some constant  $C>1$ 
$$ \frac{C\sum p_i^{\a-1} }{\sqrt{n (\sum p_i^{2\a-1} -(\sum p_i^\a)^2) }}\ge \frac{m}{\sqrt{n}} \frac{(\max_i p_i)^{\a-1}}{m^{1-\a}}\ge  \frac{m}{\sqrt{n}}$$   for large $m,n$ and   \eqref{eq:pc} implies  \eqref{eq:ld} with $\tau=1/2$. Similarly, (possibly for different $C>1$) 
$$ \frac{\sum p_i^{\a-1} }{\sqrt{n (\sum p_i^{2\a-1} -(\sum p_i^\a)^2) }}\le  \frac{C m}{\sqrt{n}} \frac{(\min_i p_i)^{\a-1}}{m^{1-\a}}\le C \frac{m}{\sqrt{n}} $$
and therefore  in this case \eqref{eq:pc} is seen to be actually equivalent to \eqref{eq:ld} with $\tau=1/2$.
\end{remark}
\begin{remark}[Plug-in Bias]\label{rem2a} Note that, in view of Jensen's inequality applied to the  strictly concave function $x\rightarrow x^\a$ for $x>0$ and  $0<\alpha<1$, we have 
$ E\S(\bhp)/\S(\bp)\le 1$. This and  the assertion $(i)$ above imply together that under the assumptions of Theorem~1 the {\em relative bias}  of  $\S(\bhp)$ satisfies   $E\S(\bhp)/\S(\bp)-1\to~0$
as $n,m\to \infty$. The standard  inequality  $\log x\le x-1$ valid for $x>0$ implies then that 
the bias of the plug-in entropy estimate satisfies \begin{equation}\label{eq:bias} E \H(\bhp)-\H(\bp)\to 0\qquad \text{as}\quad n,m\to\infty. \end{equation} \end{remark}
Unfortunately,  as may be  seen from the proof of Theorem~1 in  the appendix,  a more careful analysis of the tail events for  the plug-in estimate  than the one currently performed   is needed in order to actually establish  a convergence rate in \eqref{eq:bias}.

Turning now to our second result, note that the relation \eqref{eq:hud} suggests that  CLT of Theorem~1 could be also extended to the Renyi divergence.  The proof is again based on the Taylor expansion method where now the projection variable \eqref{eq:w} is  replaced by \eqref{eq:v}. 
\begin{theorem}[Renyi Divergence CLT]\label{thm:2} Let $V_n^{(\a)}$ be a sequence of random variables defined by \eqref{eq:v} with $\inf_n\C (V_n^{(\a)})>0$ and let 
\begin{equation}\label{eq:pc2}\left( \sum \left({q_i}/{p_i}\right)^{1-\a} +\sum \left({p_i}/{q_i}\right)^{\a}\right) (nVar V_n^{(\a)})^{-1/2}\to 0\quad\text{for}\ m,n\to \infty.\end{equation}
Then, under the LD condition \eqref{eq:ld}, as $m,n\to \infty$
\begin{itemize}
\item[(i)] $\displaystyle{\S(\bhp,\bhq)/\S(\bp,\bq) \to 1}$ in probability, 
\item[(ii)] $\displaystyle{\sqrt{n}(\S(\bhp,\bhq)-\S(\bp,\bq))/(Var V_n^{(\a)})^{1/2}\Rightarrow N(0,1)},$
\item [(iii)] $\displaystyle{\sqrt{n}\,(\a-1)(\D(\bhp,\bhq)-\D(\bp,\bq))/\C (V_n^{(\a)})\Rightarrow N(0,1).}$
\end{itemize}
\end{theorem}

\begin{remark}[Plug-in Bias]\label{rem2b} Note that, similarly as in Remark~\ref{rem2a}, we have  
$ E\S(\bhp,\bhq)/\S(\bp,\bq)\le 1$  and, by a similar argument as before,   Theorem~2$(i)$ implies $$E\D(\bhp,\bhq)-\D(\bp,\bq)\to~0\qquad\text{as}\ n,m\to \infty.$$ \end{remark}
\begin{example}[Symmetric Divergence for Power Laws] Consider the symmetric divergence ${\cal D}_{1/2}(\bp,\bq)$  with independent marginals,  which often is the  case  of interest in  NGS applications. Note that in this situation $Var V_n^{(1/2)}=1/2-(\sum \sqrt{p_i q_i})^2/2$. Suppose additionally that   $p_i=H^{-1}(\beta_1,m)/(i^{\beta_1}l_1(i))$ and  $q_i=H^{-1}(\beta_2,m)/(i^{\beta_2}l_2(i))$, $(i=1,\ldots,m)$ where the notation is as in Example~2 with $0<\beta_1\ne\beta_2<1$. Then $$Var V_n^{(1/2)}\sim \frac{1}{2}-\frac{\sqrt{(1-\b_1)(1-\b_2)}}{2-\b_1-\b_2}$$ and, consequently,  \eqref{eq:pc2} is seen as  equivalent to $m/\sqrt{n}\to 0$  (cf. also Remark~\ref{rem2} above). 
\end{example}

 With some additional effort, the two CLT results of this section may be extended to degenerate projections. This is discussed in the next section. 
 
\subsection{CLTs for Degenerate Projections }
In case of  a  degenerate projection,    the linear term of the power sum Taylor's expansion disappears (cf. formula 
(B.6) in the appendix) and  the condition \eqref{eq:pc} is no longer needed. However,  the  LD  assumption \eqref{eq:ld} has  to be slightly strengthened in order to establish the asymptotic results  for the leading (quadratic) term of the appropriate expansion. 
\subsubsection{Chi-Square Statistic CLT}
The following lemma describing  the chi-square statistic CLT  may be of independent interest for models of  sparse contingency tables. For a recent discussion of a normal approximation to the chi-square statistic in such settings, see, e.g.,  \cite{chi2approx2013}.  Here  we  apply the chi-square  CLT  formulated below     to obtain   weak limits  for the quadratic  terms in the  entropy and divergence Taylor's expansions  leading to Theorems~3 and 4 described in the next subsection. To begin, consider    
a pair of distributions $(\bp, \bq)$  and a set of positive weights ${\bm r}=(r_i)_{i=1}^m$ and  define the corresponding  chi-square ($\chi^2$) distance function as   $$\ch_{\bm r}(\bp,\bq)=n \sum \frac{(p_i-q_i)^2}{r_i}.$$
Note that,  for instance,   the $\chi^2$-distance  statistic between the empirical marginals $(\bhp,\bhq)$ is  obtained by setting  $r_i=p_i+q_i$ 
$$\ch_{\bm r}(\bhp,\bhq)=n\sum \frac{(\p_i-\q_i)^2}{p_i+q_i}$$ 
and  the  Pearson $\chi^2$-statistic  is  obtained by setting $r_i=p_i$ 
\begin{equation}\label{eq:ch1}
\chp(\bhp,\bp)=n\sum\frac{(\p_i-p_i)^2}{p_i}.
\end{equation} 
  Below we denote   $\ch_\bu(\bhu,\bu)=:\ch_\bu$.

\begin{lemma}\label{lem:ch1}   Let $(p_{ij})_{i,j=1}^m$ be the bivariate distribution of $Z=(X,Y)$  with  $X$ and $Y$  having   marginals  $(p_i)_{i=1}^m$ and $(q_i)_{i=1}^m$  where $p_i=q_i>0$. Assume $m\to \infty $ as $n\to \infty $ and 
\begin{align}
& (mn)^{-1}\sum \max (p_i^{-1},p_i^{-2}m^{-1})\to 0,  \label{eq:con2}
\end{align}
Then  as  $n\to \infty$
\begin{itemize}
\item[(i)] $\displaystyle {\frac{\ch_\bp(\bhp,\bp)-m}{\sqrt{2m}}\Rightarrow N(0,1)},$ \\  and if additionally 
\begin{equation}\label{eq:bd}
\sup_n \max_{ij} \frac{p_{ij}}{p_ip_j}=B<\infty 
\end{equation} then also \vspace{.071in}
\item[(ii)] $\displaystyle {\frac{\ch_{2\bp}(\bhp,\bhq)-\mu_n }{\sqrt{2}\gamma_n}\Rightarrow N(0,1)},$ \\ where \begin{align} \mu_n &=\sum_i (1-p_{ii}/p_i)\nonumber \\  \gamma_n^2 &=\sum_i \frac{(p_i-p_{ii})^2}{p_i^2}+\sum_{1\le i\ne j\le m} \frac{(p_{ij}+p_{ji})^2}{4p_ip_j}.\label{eq:sn}\end{align}
\end{itemize}
\end{lemma}

 \begin{remark}\label{rem3}  Note that for $\chu$ the condition \eqref{eq:con2} simplifies  to $m/n\to 0$.
 \end{remark}
 
  \begin{remark}\label{rem4}  Note that under the assumption \eqref{eq:bd} we have $m -2B\le \gamma_n^2\le m+B^2$ and therefore $\gamma_n^2\sim m$. In particular,  if $p_{ij}=p_ip_j$ then  $\mu_n=\gamma_n^2=m-1$.
\end{remark} The proof of the result may be found in the appendix. 
Its application  is  discussed next.  
\subsubsection{Pure Noise and Equal Marginals CLTs}
  The first result  covers  the case of Renyi entropy  when $\bp=\bu$. The proof is  outlined in the appendix. Recall that for real $a$ and integer $k$ we define $\binom{a}{k}=a(a-1)\cdots (a-k+1)/k!$ 
\begin{theorem}[Uniform Entropy CLT]\label{u1} Assume $m\to \infty $ as $n\to \infty $ and  $m^2/n =o (n^{-\tau})$ for $\tau>0$. Then \begin{itemize}
\item[(i)] $ \frac{n\binom{\a}{2}^{-1}[m^{\a-1}\S(\bhu)-1]-m}{\sqrt{2m}}\Rightarrow N(0,1)$
\item[(ii)] $   \frac{n[\H(\bhu) -\log m- (1-\a)^{-1}{\log(1+\binom{\a}{2}\frac{m}{n}})]}{\a\sqrt{m/2}}\Rightarrow N(0,1).$
\end{itemize}
\end{theorem}

Our second CLT result is  the following theorem for Renyi divergence  when  $\bp=\bq$. The proof is  again deferred to  the appendix.

 \begin{theorem}[Degenerate Divergence  CLT]\label{u2} Let $(p_{ij})_{i,j=1}^m$ be the bivariate distribution of $Z=(X,Y)$  with  $X$ and $Y$  having    marginals  $\bp=(p_i)_{i=1}^m$ and $\bq=(q_i)_{i=1}^m$  where $p_i=q_i>0$. Let $\mu_n$ and $\gamma_n^2$ be given by \eqref{eq:sn}. Assume $m\to \infty $ as $n\to \infty $ and that \eqref{eq:bd} holds,  as well as that 
\begin{equation}\label{eq:cond3}  \max \left\{\frac{1}{nm\min p_i^2},\frac{m}{n\min p_i}\right\}=o(n^{-\tau}). 
\end{equation} Then \begin{itemize}
\item[(i)] $ \frac{n(\a(\a-1))^{-1}[\S(\bhp,\bhq)-1]-\mu_n}{\sqrt{2}\gamma_n}\Rightarrow N(0,1)$
\item[(ii)] $   \frac{n[\D(\bhp,\bhq)- (\a-1)^{-1}\log(1+\a(\a-1)\frac{\mu_n}{n})]}{\a\sqrt{2}\gamma_n}\Rightarrow N(0,1).$
\end{itemize}
\end{theorem}
\begin{remark} Note that for $\bp=\bq=\bu$ the condition \eqref{eq:cond3} reduces to $m^2/n=o(n^{-\tau})$ required in  Theorem~3.

\end{remark}

\subsubsection{Random Sample Size}
When analyzing  NGS  data   some  part of the sequences reads   is frequently removed for technical reasons, for instance, due to poor amplification or reading  errors (see next section). In such cases  one effectively  deals with  a molecular  sample of random size.    Our  CLT   results derived  earlier may be extended to this case as well, with  the help of following simple result described in Theorem~5 below. Its  various versions have been discussed, for instance,  in the context of random allocations \citep[see, e.g.,][]{Kolchin78}. 

\begin{theorem}[Randomized  Sample CLT]\label{rcpt}  Let $(Z_n)_{n=1}^\infty$ be  a  sequence of bivariate variables    supported on an $m_n\times m_n$ integer lattice with distribution  $(p_{ij})_{i,j=1}^{m_n}$.   Let 
$(\hat{Z}_n)=(\hat{p}_{ij})_{i,j=1}^{m_n}$ $(n=1,2,3, \ldots,)$ be the sequence of the  empirical estimates, each  based on an iid sample of (deterministic) size $n$. Suppose that   the statistic  $\G_n=\G_n(\hat{p}_{ij})$  satisfies   $b_n(\G_n-a_n)\Rightarrow N(0,1)$ as $n\to \infty$ with some non-random  $(a_n, b_n)$.   Let    $(\nu_n)_{n=1}^\infty$ be a sequence of random variables independent of $(\hat{Z}_n)_{n=1}^\infty$  and  following  the binomial distributions $bin(n,\tau_n)$ with $0<\inf_n \tau_n\le \sup_n \tau_n<1$.  Then  also  $$b_{\nu_n}(\G_{\nu_n}-a_{\nu_n})\Rightarrow N(0,1).$$
  \end{theorem}
  \begin{proof}
  Denote by  $\G_{n_k}$ the random variable $\G_{\nu_k}$   conditional on the event  $\nu_k=n_k$ and by  $\Phi$  the distribution function of the standard normal random variable. 
 By assumption,  for  any  real $x$   we have $ P(\G_{n_k}\le x) \to \Phi(x)$ provided that $n_k\to \infty$ as $k\to\infty$.
 Let $\ve>0$ be sufficiently small and define $C_\ve(k_0)=\{n_k:  k(\tau_k-\ve)\le n_k\le k(\tau_k+\ve), k>k_0\}$. Note that  by the weak law of large numbers  $P(\nu_k \in C_\ve(k_0))\to 1$ as $k_0\to \infty$. Therefore 
\begin{align*} 
P(\G_{\nu_k}\le x, \nu_k\in C_\ve(k_0))  &= \sum_{n_k\in C_\ve(k_0)} P(\G_{n_k}\le x) P(\nu_k=n_k) \\
&= (\Phi(x)+\delta(k_0)) P(\nu_k \in C_\ve(k_0))
\end{align*} where $\delta(k_0)\to 0$ as $k_0\to \infty$. Accordingly,  as $k_0\to \infty$ the left-hand side converges to $\lim_k P(\G_{\nu_k}\le x)$ and  the right-hand side to $\Phi(x)$ and the result follows.
   \end{proof}

\section{Examples and NGS Applications}
We start by   providing some  numerical examples illustrating   that,  in general, the CLT results discussed above do not  hold without  assumptions on the relative rate of $m$ and  $n$. Next, we  show  two examples of  applicability of our results  to analyzing biodiversity of NGS data.  The first one is concerned with  comparing the diversity of T-cell receptor populations in transgenics  mice, whereas the second one aims at identifying the hepatocellular carcinoma  transcription profiles  in humans.  For the purpose of the   T-cell receptors   example, we  propose  a  sequential statistical procedure of NGS signal filtering based on our CLT results from the previous sections. We begin by  pointing out  to some subtleties   in   the CLT results discussed in Section~3.
\subsection{Power Law and Pure Noise Models}
Consider the power law model from Example~2 in Section~\ref{sec:ld} with $\beta=1$ and $l(x)\equiv 1$. Note that  in this case 
$ (n \min_i p_i)^{-1} \sim m \log m/\,n$ as well as  $\sum p_i^{\alpha-1}(n VarW_n^{(\alpha)})^{-1/2} \sim O(m(\log^{2\a} m/n)^{1/2})$ and therefore the assumptions of Theorem~1 are satisfied as soon as 
\begin{equation}\label{eq:10}
n^{\tau-1} m\to 0
\end{equation} for some $\tau>1/2$. Similarly, the assumption \eqref{eq:con2} of Lemma~2 is satisfied as soon as  
\begin{equation}\label{eq:11}
\log^2\!m\,\frac{m}{n}\to 0.
\end{equation}
In Figure~\ref{fig1} we illustrate the convergence results of Theorem~1$(iii)$ and Lemma~2$(i)$ for   this power law model and $\alpha=0.5$.  The panels of Figure~\ref{fig1} presents the  sample vs standard normal quantile  (QQ) plots  for the normalized Renyi entropy statistic and  the normalized Pearson statistic \eqref{eq:ch1} based on  $B=5000$ samples from the power law distribution, each  with $m=1000$ and three different values of $n=m^{1+\ve}$ ($\ve=-0.5, 0.5, 1.5$). As seen from the plots, in the absence of  \eqref{eq:10} the CLT result for the Renyi entropy (cf. Theorem~1$(iii)$) does not hold. Moreover, the middle panel QQ plot  indicates that  for  large $m,n$  satisfying $n=m^{3/2}$ the discrepancy between distribution of  the entropy function  and its plug-in estimate appears in a form of deterministic shift, indicating the presence of substantial asymptotic bias and hence the lack of  convergence  \eqref{eq:bias}. Similarly,  when  \eqref{eq:11} is not satisfied than   the Pearson statistic CLT given in Lemma~2$(i)$ fails with the middle panel again indicating that the bias of the estimate does not vanish when $m$ is too large relative to $n$. 
\begin{figure}[htp]
 \centering{
\includegraphics[scale=0.6,trim=0 10 0 10]{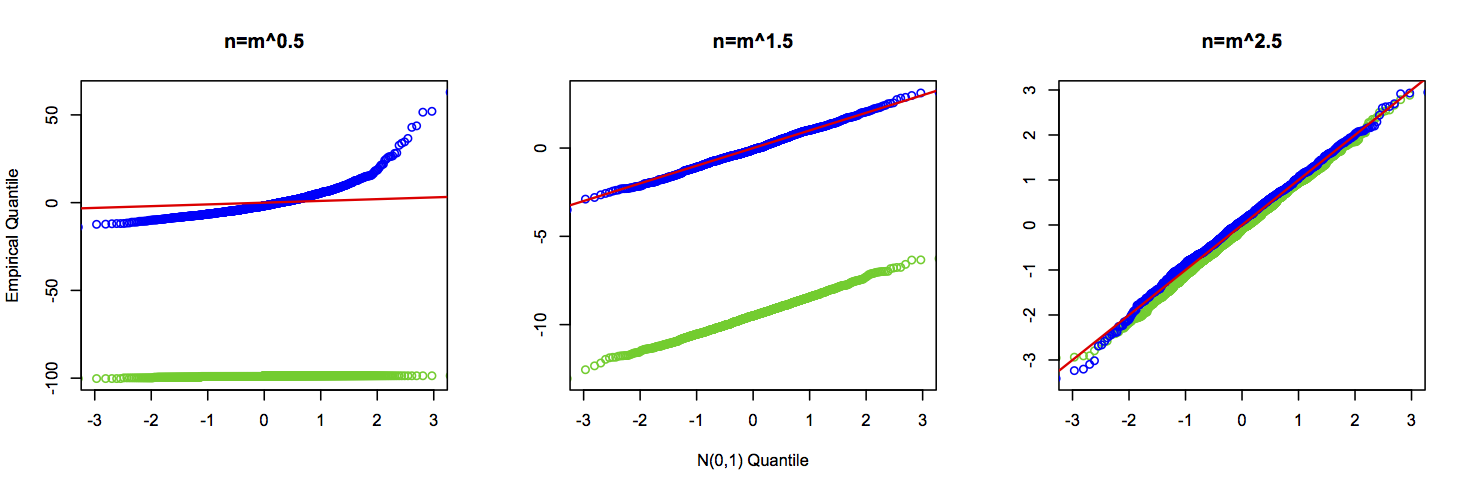}}
\caption{{\bf Projection CLTs.} Normal QQ  plots  for  the normalized Renyi entropy (Theorem~\ref{thm:1}$(iii)$, lower (green) curve) and normalized Pearson $\chi^2$ statistic (Lemma~2$(i)$,  upper (blue) curve)  for the power law distribution $p_i=1/i$. The panels shows quantile plots  with different values of $n=m^{1+\ve}$ ($\ve=-0.5, 0.5, 1.5$) and $m=1000$. The   solid  (red) line gives  quantiles of the standard normal distribution for reference. }\label{fig1}
\end{figure}

For comparison,  we also considered the uniform distribution (pure noise) model $p_i=1/m$. Note that it  may be  viewed as a degenerate power law where $\beta=0$ and $l(x)\equiv1$.  Recall that according to Theorem~3 $(ii)$ and Lemma~2 $(i)$, the sufficient conditions for the respective CLTs are  $m^2/n^{1-\tau}\to 0$ and $m/n\to 0$ (see Remark~\ref{rem3} for the latter one).  The necessity  of  these conditions is  illustrated in the panels of  Figure~\ref{fig2}  where we again present the (normal) QQ plots for the Renyi ($\alpha=0.5$) and the Pearson statistics   for the same values of $B, n$ and $m$ as in Figure~\ref{fig1}. As seen from these plots,  only in the last panel, when $m^2/n\approx  0$, we get good  CLT approximation for both statistics. These  results appear  consistent with our theoretical results from Theorem~3 and Lemma~2.
\begin{figure}[htp]
 \centering{
\includegraphics[ scale=0.6,trim=0 10 0 10]{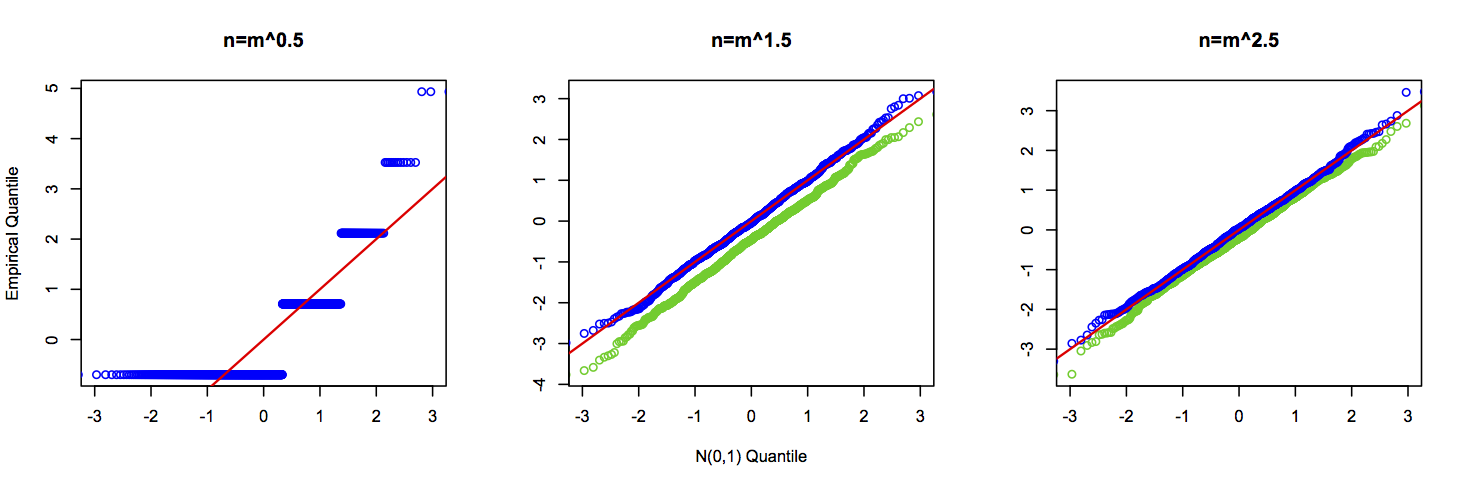}}
\caption{{\bf  Degenerate Projection CLTs.} Normal QQ plots for  the normalized uniform Renyi entropy (Theorem~\ref{u1}$(ii)$, represented by the lower  (green) curve)  and the normalized  Pearson $\chi^2$-statistic (Lemma~2 $(i)$, represented by the upper (blue) curve) with $p_i=m^{-1}$. The panels shows quantile plots with  different values of $n=m^{1+\ve}$ ($\ve=-0.5, 0.5, 1.5$) and $m=1000$. The   solid  (red) line gives  the quantiles of the standard normal distribution for reference.  Note that the normalized Renyi entropy is undefined for   the first panel. }\label{fig2}
\end{figure}

Although not presented here due to space considerations, similar examples based on the  bivariate power laws  may be used to  illustrate  the necessity of the assumptions of type  \eqref{eq:pc2} and \eqref{eq:cond3} in the CLT results for divergence in Theorems~2$(iii)$ and 4$(ii)$. 

\subsection{Applications to NGS Data}
 Our CLT results described in Section~\ref{sec:lims} were originally motivated by questions rising in NGS data analysis. Below  we describe   two  examples which adhere to the following basic framework.   Denote by  $\be_1,\be_2$  two independent noise distributions  each on $m$ support points,    and assume   that a pair $(\bp, \bq)$ of marginal distributions may be represented as \begin{equation}\label{eq:nm}(\bp,\bq) =\la (\btp, \btq)+(1-\la)(\be_1,\be_2)\end{equation} where $(\btp,\btq)$ is a pair of  marginal   distributions having no common support points with  $(\be_1,\be_2)$ and $\la$ is the mixing proportion (or prior probability of signal).   We  assume that  each $\be$ is a simple finite mixture of  $K$  uniform distributions on separate support. Note that  the noise-and-signal model from Example~\ref{ns1} in Section~\ref{sec:ld} may be viewed as a (univariate) special case of \eqref{eq:nm} with $K=1$.  In the  first  example below we took $K=2$.   
 \begin{table}
\begin{center}
\begin{tabular}{ccc}
& Antibiotic ($\btp$)& Control  ($\btq$)\\
\hline
\tabT$n$ & 39,084&39,084\\
$m$ & 165&165\\
$k_m$ & 17 & 17 \\
$\hat{\la}$ &0.46 &0.46\\
$\hat{\beta}$ & 0.869(0.05) & 0.971(0.05) \\
${\cal H}_{1/2}$ &4.81 (4.79, 4.82) & 4.64 (4.63, 4.67) \\
$ENC_{1/2}$ & 122.73 (120.30, 123.97) & 103.54 (102.51, 106.70)\\
${\cal D}_{1/2}$ & \multicolumn{2}{c}{0.155 (0.147, 0.163)}  \\
\hline

\end{tabular}
\caption{{\bf Results of  TCR data analysis}. The mixture model \eqref{eq:nm} with   heavy-tailed power laws  fitted to two sets of TCR  counts derived from mouse MLN before and after an antibiotic treatment as described in \cite{Cebula2013}. The empirical Renyi entropy, the Hill number and the Renyi diversity CIs (in parenthesis)  are obtained from the CLT results of Theorems~1 and 2.}\label{tab1}
\end{center}
\end{table}

\vspace{.1in}\noindent {\bf Algorithm 1(NGS  Diversity Analysis with $\D$ or $\S$)}\vspace{-.1in}
\begin{cenum}
\item {\em Exponent ($\a$) selection}. Use  problem-specific criteria (e.g. sample coverage, see \cite{RS12}) to identify the appropriate $\alpha$ value. If no  prior knowledge  exist,  the value  $\a=1/2$  (the  Bhattacharyya distance) may be  often used. 
\item {\em Noise filtering}. Identify the number of mixture components $K$ and the cut-off count(s) $k_m$ for  the support of  $\be_i$ in \eqref{eq:nm} with   a sequential (starting from the lowest empirical frequency) procedure  based on  Lemma~2$(i)$ with  $\bp=\be_i$ ($i=1,2$).    The values of $\la$ is then estimated as the proportion of a sample falling into the $m$ 'noise' categories.  
\item{\em Equality testing}. For a pre-determined value of $\a$,  test the hypothesis $H_0: \btp=\btq$ by  comparing the  observed value of $\D$  (alternatively, $\S$) with  the asymptotic normal distribution in  Theorem~4. 
\item{\em Difference quantification}. If $H_0$ is  not rejected,  conclude that $\D\equiv 0$ ($\S\equiv 1$). Otherwise,   apply Theorem~2 to obtain confidence bounds 
for $\D$ ($\S$).\end{cenum}

\subsubsection{T-Cell Receptor Populations} In this  example we apply Algorithm~1 to  measure similarity between   a pair of  T-cell receptor (TCR) populations based on the observed  NGS counts of  receptor-specific nucleotide sequences.   With the  current NGS  technology,  the two main difficulties in  comparing  TCR populations are  to adjust the under-sampling  bias due to unobserved rare types  and  the `ghost`  types  created   due to  the sequencing errors  \citep{Wang:2014aa}.    The first  problem may be often alleviated  by   applying  diversity criteria, like the Renyi entropy and divergence, which allow for the sample-based up-weighting of  rare  counts (see  \citealt{RS12}).   The second one  requires typically additional assumptions, in order to perform  analysis as outlined  in Algorithm~1$(ii)$.  A recent  detailed overview  of the   TCR diversity analysis  methods was presented  by   \cite{RS12} and  earlier on, in a more general context of biodiversity,  by  \cite{Hsieh:2006sy} and  \cite{Magurran:2005by}.   For illustration, we analyze here  two populations derived from the mesenteric lymph nodes (MLN) of a  TCR mini-mouse before and after an antibiotic treatment.  The details of the experiments and   a dataset description are  given in   \cite{Cebula2013}.  For the current analysis it is important to note that,  since the experimental groups  consisted of  different animals,  we may consider  two experimental groups  as independent.  The total combined  sample size (or sequencing depths)  was $n=72,030$,  with initial $m_0=6,336$ receptor types. After performing  step $(ii)$ of   Algorithm~1 $m=165$  types were  identified  as ``signal"  based on the  cut-off  $k_m=17$ in both populations. The signal population corresponded to the remaining sample size of $38,896$   or about  54\% of the original NGS counts. We used  $\D$ with  $\alpha=1/2$ as the diversity measure  in step $(iii)$-$(iv)$  of Algorithm~1.   Based on Theorem~2,  the asymptotic  $P$-value for testing $H_0: \btp=\btq$ was found to  be less than $10^{-4}$ and hence the hypothesis of equal diversity of the two populations was rejected (see Algorithm~1$(iii)$).  

\begin{figure}[htp!] \centering{ \
\includegraphics[scale=0.4,trim=15 15 15 15]{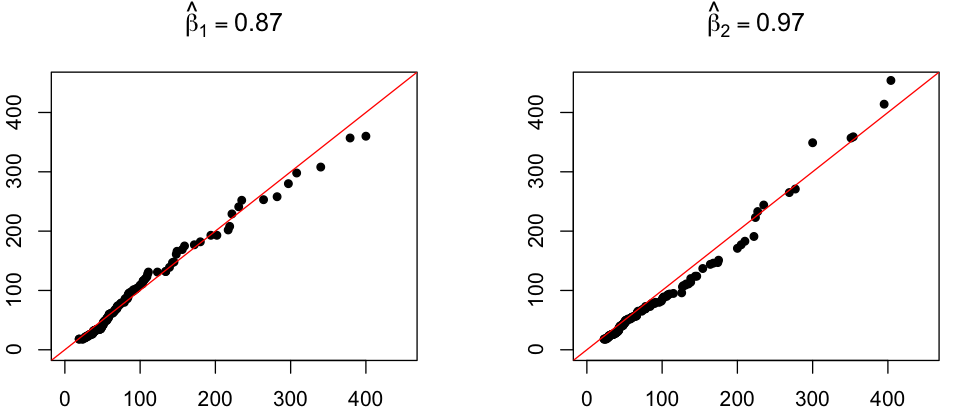}}
\caption{{\bf Power law fit for TCR data}. QQ plot of the TCR data against quantiles of a power law distribution with $\beta_1=0.87$ ($SE=.05$)  and $\beta_2=0.97$ ($SE=.05$) values fitted via  the least squares method. }\label{fig:qq} 
\end{figure}

To compare this finding with a more standard  parametric analysis, we additionally  fitted, with the least squares method, the counts of $165$ receptor types in two populations to the   power law distributions.  Since the respective  exponent  values  for the two  fitted populations were found  to be different, with    $\beta_1=.87$ (for antibiotic treated mice) and $\beta_2=.97$ (for untreated), the parametric analysis confirmed the findings of Algorithm~1.  For illustration,  the plots of the fitted power law quantiles versus  the empirical ones  are presented in Figure~\ref{fig:qq}.  Additionally, the diversity of each of the TCR populations in terms of its  respective Renyi entropy ${\cal H}_{1/2}$ and the Hill number $ENC_{1/2}$ as well as the diversity difference  measured by the Renyi divergence  ${\cal D}_{1/2}$  are listed  in Table~\ref{tab1}, along with the  corresponding asymptotic  confidence intervals  obtained via Theorems~1 and 2.  
 As seen from the values in Table~\ref{tab1},  although the diversity of each of the NGS populations was relatively similar  in terms of the two populations count patterns, it differed in terms of the specific TCR types expressed.

\subsubsection{Gene  Expression Profiling}

Beyond Algorithm~1, the results of Section~3   may be applied to facilitate various other biodiversity analysis,   for instance, in    simultaneous comparison of  several pairs of  molecular  samples.  We  illustrate this  with an NGS   data example from  the  recent hepatocellular carcinoma (HCC)  study in  \cite{HCC14} which we obtained  through the  gene expression omnibus (GEO)  database. The GEO dataset consists of HCC tumor-infected ($T$) and healthy liver  ($N$)  tissue samples from three individuals denoted below  as follows  in relation to their original database designations   $T1=HCC448T, T2=HCC473T, T3=HCC510T$  and  $N1=HCC448N, N2=HCC473N, N3=HCC510N$.
For this dataset one of the questions  of research interest was whether the expression profiles of  genes associated with regulation of cell proliferation and programmed cell death  differ  across $T$ and $N$ samples as well as across individuals  (cf., e.g., \citealt{Kong:2013aa}). To address this specific  question, in contrast with the  previous TCR example,  we were  thus only interested in  a pre-selected subset of  the NGS counts. The final values of  $m=1332 $ and $n$ between 1.2 and 1.9 million reads \footnote{Based on these  values, the empirical versions of the conditions for the relevant  theorems in Section~3 were  considered satisfied.}  were obtained after aligning the   pre-selected NGS fragments   to the  HG19  reference genome with the Tophat2/Bowtie2 software \citep{Kim:2013aa} and performing the  transcript annotation with the  Ensembl genome browser (\url{www.ensembl.org}).  After  the final fragments-to-counts  conversion, our  data analysis was performed in three steps. 
First, the null hypothesis of the tissue homogeneity $H_0^{all}= \{T_1=N_1=T_2=N_2=T_3=N_3\}$ was tested (and rejected) based on the result of  Theorem~4 and the  corresponding asymptotic $p$-value obtained from the $\chi^2(3)$ distribution. 
Next,   the  hypothesis of the across-individuals homogeneity was tested by evaluating three pairwise null hypothesis $H_0^{ij}=\{{\cal D}_{1/2}(T_i,N_i)={\cal D}_{1/2}(T_j,N_j)]\},$ $1\le i<j\le 3$  (each  rejected) based on Theorem~4.  Finally, having rejected the homogeneity hypothesis  we have used the result of Theorem~2 to quantify the differences between the three sets of  $T$ and $N$ tissue samples. The details of the analysis are presented in Table~\ref{tab3}. As seen from the numerical results,   it seems that despite the large  individual differences between patients, the set of $m=1332$ genes   associated with cell proliferation and  death may be used to distinguish between  T-type  and N-type samples in HCC patients.

\begin{table}
\begin{center}
\begin{tabular}{cccccc}
 Hypothesis  &Statistic   & $P$-value & ${\cal D}_{1/2}$ Value (CI)   \\
\hline
\tabT $H_0^{all}$&$\sum w_i[{\cal D}_{1/2}(i)-\mu_i]^2$  & $< 0.001$  & NA \\
\tabT  $H_0^{1,2}$  &${\cal D}_{1/2}(1)-{\cal D}_{1/2}(2)$& $< 0.01$ & ${\cal D}_{1/2}(1)$=0.553 (0.551, 0.555)  \\
\tabT  $H_0^{2,3}$   &${\cal D}_{1/2}(2)-{\cal D}_{1/2}(3)$ & &${\cal D}_{1/2}(2)$=0.292 (0.291, 0.294) \\
\tabT  $H_0^{3,1}$  & ${\cal D}_{1/2}(3)-{\cal D}_{1/2}(1)$ &  &${\cal D}_{1/2}(3)$= 0.346 (0.345 0.348)\\
\tabB\\
\hline
\end{tabular}
\caption{The $95\%$ confidence intervals for the pairwise symmetric Renyi Divergence ${\cal D}_{1/2}$ between the tumor and control (healthy) tissues from three individuals based on the profile of expression of pre-selected  $m=1332$ transcripts  related to cell proliferation. Here ${\cal D}_{1/2}(i)$ denotes  ${\cal D}_{1/2}(T_i,N_i).$}\label{tab3}
\end{center}
\end{table}

\section{Summary and Conclusions}

  We derived  two sets of  limit theorems for the Renyi entropy and divergence statistics. The first set of results  holds  for lineralizeable statistics (their first order Taylor approximations exist) whereas the second one holds in the degenerate case (when  the first order approximations  vanish)  and requires   analyzing the quadratic terms in the Taylor expansions.   Our  Renyi entropy limit theorems   complement  those obtained elsewhere for the Shannon entropy and divergence.  
  
 Based on the CLT results we have proposed here  a new framework for  analyzing molecular diversity of molecular (especially NGS) data  based on  the idea of  analyzing the   frequency/contingency tables  where cell counts are highly unbalanced (for instance, as arriving from  mixtures   of  heavy tailed, power-law type and uniform distributions) and the number of cells or, equivalently,  the counts distribution support size $m$, increases with the sample size $n$. For analyzing such tables,   we suggested using the  empirical Renyi entropy and divergence as  the statistical measures of, respectively,  diversity and pairwise similarity of different molecular sub-populations.

 In  the two  examples of NGS analysis we have  shown how  the Renyi entropy  methods may be used for filtering out low frequency noise and for establishing valid confidence bounds in  pairwise divergence analysis for pre-selected transcripts. However,  it was also seen that  in order to apply our CLT results  the number of  transcripts had to be small relative to the sequencing depth.  For the special class of  heavy-tailed  power law distributions, our  results in particular  indicate that the appropriate entropy CLTs  are valid  (and thus so is our  proposed analysis framework)  when, roughly speaking,  $m/\sqrt{n}\to 0$ and not otherwise.  As such  restriction may be often limiting in very high diversity NGS data,   other  statistics  beyond  those  discussed here  and  not requiring such  condition  could be also  of interest.  We hope to  pursuing  this matter  further in our future 
work.

%

 

 \bibliography{bmb2}

\newpage
\appendix
\renewcommand{\theequation}{\thesection. \arabic{equation}}
\numberwithin{equation}{section}
\setcounter{page}{1}
\vspace{40pt}\noindent {\bf\large Appendix}
\section{Proofs for Non-Degenerate Projections  (Section 3.1)}
\subsection*{Auxiliary Results}

First, we establish the following simple result on   binomial moments. 

\begin{lemma}[Binomial moment bound] \label{lem:a1} Let $[x]$ denote the largest integer smaller or equal to $x$ and let $\p_n$ be an empirical binomial proportion  from $n$ independent Bernoulli trials with the success probability $0<p_n<1$. Assume $np_n\to \infty$ as $n\to \infty$. Then for any integer $d\ge1$ and sufficiently large $n$
$$\vert E \left( \p_n n-p_nn\right)^d\vert\le C_d\, (np_n)^{[d/2]} $$ for some universal ($n$ free) constant $C_d$.
\end{lemma}

\begin{proof} Let $X$ be a binomial $Bin(n,p_n)$ random variable and set $\mu =n p_n$.  
Then (see e.g, \cite{Knoblauch08}) $$ E(X-\mu)^d = \sum_{i=0}^d \binom{d}{i} (-\mu)^{d-i} EX^i=\sum_{i=0}^d \binom{d}{i} (-\mu)^{d-i} \sum_{k=0}^i \stir{i}{k}p_n^k n^{\underline{k}}, $$ where $\stir{i}{k}$ denotes a  Stirling number  of the second kind (i.e. the number of ways to  partition  a set of $i$ objects into $k$ non-empty subsets) and $n^{\underline{k}}=n(n-1)\cdots(n-k+1)$.
Let $$c_{d,k}= \sum_{i=d-k}^d (-1)^{d-i} \binom{d}{i}\stir{i}{i-d+k}=\sum_{i=0}^k (-1)^i \binom{d}{i} \stir{d-i}{k-i}$$ denote the coefficient at $\mu^k$ in the expression for  $E(X-\mu)^d$.
Then for $1\le k\le d$\begin{equation}\label{eq:rec} c_{d+1,k}= d\,c_{d-1,k-1}+k\,c_{d,k}.\end{equation}
Indeed, 
\begin{align*}
&d\,c_{d-1,k-1}+k\,c_{d,k}\\
 &= d \sum_{i=0}^{k-1} (-1)^i \binom{d-1}{i} \stir{d-i-1}{k-i-1}+k \sum_{i=0}^k (-1)^i \binom{d}{i} \stir{d-i}{k-i}\\
&=  \sum_{i=0}^{k-1} (-1)^i \binom{d}{i+1} (i+1) \stir{d-i-1}{k-i-1}+k \sum_{i=1}^k (-1)^i \binom{d}{i} \stir{d-i}{k-i}+k\,\stir{d}{k}\\
&= \sum_{i=1}^{k-1} (-1)^i \binom{d}{i}  \left((k-i)\stir{d-i}{k-i}+\stir{d-i}{k-i-1}\right)+k\,\stir{d}{k}-\sum_{i=1}^{k-1} (-1)^i \binom{d}{i}  \stir{d-i}{k-i-1}\\ &\text{and,  using  the recursions for the Stirling numbers and the  binomial coefficients,}\\
&=\sum_{i=1}^{k-1} (-1)^i \left(\binom{d+1}{i}-\binom{d}{i-1}\right) \stir{d-i+1}{k-i}+\stir{d+1}{k}-\stir{d}{k-1}-\sum_{i=1}^{k-1} (-1)^i \binom{d}{i}  \stir{d-i}{k-i-1}\\
&=c_{d+1,k}+0.
\end{align*}
Let us argue that for any $d\ge 1$ we have \begin{equation}\label{eq:rec1}c_{d,k}=0\quad \text{for $k$ such that $d/2<k\le d$}.\end{equation}  The proof of \eqref{eq:rec1} is by induction with respect to $d\ge1$. Note that the statement   is true  for  $d=1$ due to $c_{d,d}=0$ for $d\ge 1$ (but  $c_{0,0}=1$). Now, if  $k> (d+1)/2$ then $k-1>(d-1)/2$ and $k>d/2$ and thus \eqref{eq:rec} implies $c_{d+1,k}=0$ for  $k>(d+1)/2$ since the induction assumption implies $c_{d-1,k-1}=c_{d,k}=0$.  Hence \eqref{eq:rec1} holds  and consequently  the   highest  power of  $\mu$  in the expansion of  $E(X-\mu)^d$  cannot exceed $d/2$. This yields the assertion of the lemma.

\end{proof}

\begin{lemma}\label{lem:lc}  Set $\W_n^{(\a)}=(W_n^{(\a)}-EW_n^{(\a)})/(Var W_n^{(\a)})^{1/2}$. Under  the  assumptions of Theorem~\ref{thm:1},  the Lindeberg condition \begin{equation}\label{eq:lc}\forall_{\varepsilon>0}\ E (\W_n^{(\a)})^2I(|\W_n^{(\a)}|>\varepsilon\sqrt{n})\to 0, \quad  n\to \infty 
\end{equation} is satisfied. Consequently, $$\sqrt{n}\sum_{i=1}^n \W_{ni}^{(\a)}\Rightarrow N(0,1) $$ with   $n$ iid random variables $\W_{ni}^{(\a)}$ equidistributed 
with  $\W_n^{(\a)}$. Moreover, the result  remains true if we replace above $W_n^{(\a)}$ by $V_n^{(\a)}$  under the assumptions  of Theorem~\ref{thm:2}.   \end{lemma}
\begin{proof} We shall only prove the statement for $W_n^{(\a)}$,  as the   proof  for $V_n^{(\a)}$ is similar.  
For notational convenience, set $\sigma^2_n=Var W_n^{(\a)}$, $\mu_n=E W_n^{(\a)}$ and 
$\W_n=\W_n^{(\a)}$. In view of  \eqref{eq:pc} we have  as $n\to \infty$
\begin{equation}\label{eq:maj}
\frac{\mu_n}{\sqrt{n}\sigma_n}=\frac{\a\sum p_i^\a}{\sqrt{n}\sigma_n}\le \frac{\sum p_i^{\a-1}}{\sqrt{n}\sigma_n}\to 0.
\end{equation}
Note that $\W_n=\a\sigma_n^{-1}\sum p_i^{\a-1}(\delta_i-p_i) $ where the vector 
$(\delta_1,\ldots,\delta_m)$ represents  a single trial  multinomial random vector with parameters $(p_1,\ldots,p_m)$. For any $\ve >0$ 
\begin{align}\label{eq:bnd}
E \W_n^2I(|\W_n|>\ve \sqrt{n}) &= \a^2\sigma_n^{-2} E(\sum p_i^{\a-1}(\delta_i-p_i))^2I(|\W_n|>\ve \sqrt{n}) \notag\\
&\le   \a^2\sigma_n^{-2} E[(\sum p_i^\a)^2+\sum \delta_i p_i^{2(\a-1)}]I(|\W_n|>\ve \sqrt{n}).
\end{align}
Since  by \eqref{eq:maj}  $\mu_n=o((n\sigma^2_n)^{1/2})$,  then by the definition of $\delta_i$,  for sufficiently large $n$ we have  
\begin{align*}
\{\omega: |\W_n|>\ve\sqrt{n}\} &=\{\omega: \a|\sigma_n^{-1}\sum p_i^{\a-1}(\delta_i-p_i)|>\ve\sqrt{n}\}
\\ &=\{\omega: \delta_i=1 \text{ for $i$ such that } \a\vert  p_i^{\a-1}-\mu_n \vert >\ve\sqrt{n}\sigma_n\}\\ 
&\subset\{\omega: \delta_i=1 \text{ for $i$ such, that } \a p_i^{\a-1}>\frac{\ve}{2}\sqrt{n}\sigma_n\}  \}\\ 
&=:\{\omega: \delta_i=1 \text{ for } i\in J_n\}
\end{align*}
where the last equality defines the set of indices $J_n$. Note that the size of the set $J_n$ satisfies $|J_n|\to 0$ as $n\to \infty$,  due to \ $\max_{1\le i\le m_n} p_i^{\a-1}/\sqrt{n}\sigma_n\to 0$ as $n\to \infty$, which is implied by \eqref{eq:pc}.
This  and  \eqref{eq:bnd} give therefore (at least for large $n$)
\begin{align*}
E \W_n^2I(|\W_n|>\ve \sqrt{n}) &\le  \sigma_n^{-2} \sum_{i\in J_n}p_i(\mu_n^2+ \a^2p_i^{2(\a-1)})= (\mu_n/\sigma_n)^2 \sum_{i\in J_n}p_i+ \a^2\sigma_n^{-2} \sum_{i\in J_n} p_i^{2\a-1}\\
& \le  2\a(\mu_n/\sigma_n)^2  \sum p_i^\a/(\ve\sqrt{n}\sigma_n) + (\mu_n^2/\sigma_n^2 +1)  \sum_{i\in J_n} p_i^{2\a-1}/(\mu_n^2+\sigma_n^2)\to 0
\end{align*}
as $n\to \infty$,  since $\sup_n(\mu_n/\sigma_n)^2<\infty$ by the assumptions of Theorem~1 and $\a^2\sum p_i^{2\a-1}=\mu_n^2+\sigma_n^2$. The weak convergence assertion follows now by the Lindeberg central limit theorem (see, e.g, \cite{Shao04} Chapter 1).

\end{proof}

\subsection*{Proof of Theorem~\ref{thm:1}} Let us first establish part $(ii)$. 
Note that  \eqref{eq:pc} implies that \begin{equation}\label{eq:an}a_n^2=n/Var W_n^{(\a)}\to \infty, \end{equation} in view of 
\begin{equation*} a_n^2\sum p_i^{\a-1}/n\ge \a^2a_n^2\sum p_i^{2\a-1}/n\ge 1\end{equation*} 
which yields  $ a_n\ge (n Var W_n^{(\a)})^{1/2}(\sum p_i^{\a-1})^{-1}.$
By  Taylor's expansion 
\begin{equation}\label{eq:taylor} \S(\bhp)-\S(\bp)=\sum \p_i^\a-\sum p_i^\a=\sum \a p_i^{\a-1}(\p_i-p_i)+ R_{n}
\end{equation}
where  $$R_{n}=\sum\binom{\a}{2} p_i^\a\left(\frac{\p_i-p_i}{p_i}\right)^2 \left(\theta_i\left(\frac{\p_i}{p_i}-1\right)+1\right)^{\a-2}\quad \text{for some random }\theta_i\in (0,1). $$ 
Fixing   $\delta\in(0,1/2)$, for  any  $\ve>0$, we have \begin{align*}P(|R_n|>\ve) &=P\left(|R_n|>\ve,\max_i \left|\theta_i\left(\frac{\p_i}{p_i}-1\right)\right|\le\delta\right)+P\left(|R_n|>\ve,\max_i \left|\theta_i\left(\frac{\p_i}{p_i}-1\right)\right|> \delta\right)\\ & \le P\left(|R_n|>\ve,\max_i \left|\theta_i\left(\frac{\p_i}{p_i}-1\right)\right|\le\delta\right)+P\left(\max_i \left|\theta_i\left(\frac{\p_i}{p_i}-1\right)\right|> \delta\right)\\ &=:(I)+(II).\end{align*}
First, note $$ (I)\le \frac{1}{\ve}E|R_n|I\!\left(\max_i \left|\theta_i\left(\frac{\p_i}{p_i}-1\right)\right|\le\delta\right)\le \frac{1}{\ve}\binom{\a}{2}\sum p_i^\a E\left(\frac{\p_i-p_i}{p_i}\right)^2 \left(1-\delta\right)^{\a-2}\le\frac{C_1}{\varepsilon\,n} \sum p_i^{\a-1}. $$   Now, recall the condition \eqref{eq:ld} and consider $d\ge 1$ large enough so that $d\tau>1$ and hence $(np_\ast)^{-d}\le n^{-1}$ for sufficiently large $n$. Applying Bool's (subadditivity) inequality bound and Lemma~\ref{lem:a1} to  the $2d$-th central moments of the $\hat{p}_i$'s, we get 
$$ (II)\le P\left(\max_i \left|p_i^{(\a-1)/2d}\left(\frac{\p_i}{p_i}-1\right)\right|>\delta\right) \le \sum\frac{C_2\, p_i^{\a-1}}{\delta^{2d}\,n^d p_i^d} \le \frac{C_3}{\delta^{2d}\,n} \sum p_i^{\a-1}  $$  where 
$C_1,C_2, C_3$ are constants independent of $n$ and $p_i$ (with $C_2$ being $C_{2d}$ of  Lemma~\ref{lem:a1}). 
Therefore,  for any $\varepsilon>0$ and the numerical sequence $a_n=|a_n|=(n/Var W_n^{(\a)})^{1/2}\to \infty$ (cf. \eqref{eq:an}) as well as   a possibly different set of constants $C_1,C_2,C_3$
\begin{equation}\label{eq:rem}P(|a_n R_n|> \ve)\le \max\left(\frac{C_1 a_n}{\ve},\frac{C_2}{\delta^{2d}} 
\right) \sum p_i^{\a-1}/n\le C_3 a_n \sum p_i^{\a-1}/n\to 0\end{equation} as $n\to \infty$, due to 
\eqref{eq:pc}. Note that the random variable $\a\sum  p_i^{\a-1}(\p_i-p_i)$ has the same distribution 
as ${\cal W}_n=\sum (W_{ni}^{(\a)}-EW_{ni}^{(\a)})/n$, with   iid random variables $W_{ni}^{(\a)}$distributed 
as   \eqref{eq:w} and that, due to    \eqref{eq:rem},  $a_n R_n=R_n/(Var {\cal W}_n)^{1/2}=o_p(1)$. Since    Lemma~\ref{lem:lc} ensures  that $({\cal W}_n-E{\cal W}_n)/(Var {\cal W}_n)^{1/2}\Rightarrow N(0,1)$, the 
result follows.

To argue part $(i)$,  note that from the definition of \eqref{eq:w} we have  $EW_n^{(\a)}=\a^2\S(\bp)=\a^2\sum p_i^\a\ge \a^2$ and   by \eqref{eq:an} 
$\tilde{a}_n=(n\S(\bp)^2/Var W_n^{(\a)})^{1/2}\to \infty$. Consequently,  part $(i)$ follows immediately from part $(ii)$. 


Finally, we show part $(iii)$.
 Consider  arbitrary $\delta\in(0,1)$ and note that on the events $|\S(\bhp)/\S(\bp)-1|\le \delta$, by   
  Taylor's expansion  for  $|x|<1$, we have  \begin{equation}\label{eq:log}\log(1+x)=x-\frac{x^2}{2(1+\theta x)^2},\end{equation} where $\theta\in (0,1)$  and hence 
\begin{align}
(1-\a)(\H(\bhp)-\H(\bp))=\log\left( \frac{\S(\bhp)}{\S(\bp)}\right) &=\frac{\S(\bhp)}{\S(\bp)}-1-\frac{\left(\frac{\S(\bhp)}{\S(\bp)}-1\right)^2}{2(1+\theta(\frac{\S(\bhp)}{\S(\bp)}-1))^2}\notag\\ &=\frac{\S(\bhp)-\S(\bp)}{\S(\bp)}+T_n.\label{fromRtoT}
\end{align}  where the last equation defines $T_n$.
  By applying again the expansion argument used in the  proof of part $(ii)$  with  $\S(\bhp)/\S(\bp)$ in place of $\S(\bhp$) and $R_n/\S(\bp)$ in place of  $R_n$ and   the   sequence   $\tilde{a}_n=(n\S(\bp)^2/Var W_n^{(\a)})^{1/2}\to \infty$,  we see that 
  \begin{equation}\label{eq:1}
\tilde{a}_n (\S(\bhp)/\S(\bp)-1)= \sqrt{n}\,\frac{\S(\bhp)-\S(\bp)}{(Var W_n^{(\a)})^{1/2}}\Rightarrow N(0,1)
  \end{equation} as $n\to \infty.$
  Note that  for any $\ve >0$ 
  $$P(\tilde{a}_n|T_n|>\ve, |\S(\bhp)/\S(\bp)-1|\le \delta)\le P\left(\tilde{a}_n \frac{|\S(\bhp)/\S(\bp)-1|}{2(1-\delta)^2}>\ve/\delta \right)$$ and therefore 
  \begin{align}\label{eq:1b}  P(\tilde{a}_n|T_n|>\ve)\le P(|\S(\bhp)/\S(\bp)-1|> \delta) +P\left(\tilde{a}_n |\S(\bhp)/\S(\bp)-1|> 2\ve(1-\delta)^2/\delta \right).
   \end{align}  In view of $(i)$ and \eqref{eq:1}, denoting the absolute value of $N(0,1)$ by $|N|$,   \eqref{eq:1b} yields 
   $$ \limsup_n P(\tilde{a}_n|T_n|>\ve) \le 0+P(|N|\ge 2\ve(1-\delta)^2/\delta ).$$
    By taking $\delta>0$ to be sufficiently small we get $$  \limsup_n P(\tilde{a}_n|T_n|>\ve)\le \gamma $$  for  arbitrary $\gamma>0$,   and therefore  $ \limsup_n P(\tilde{a}_n|T_n|>\ve)=0$. Thus
for  any $x$, we have
  $$P\left( \sqrt{n}\,\frac{(1-\a)(\H(\bhp)-\H(\bp))}{\a\,{\cal CV} (W_n^{(\a)})} \le x\right) =P\left(\sqrt{n}\,\frac{\S(\bhp)-\S(\bp)}{(Var W_n^{(\a)})^{1/2}} \le x\right)+o_p(1)$$ and the result follows from part $(ii)$.
\qed

\subsection*{Proof of Theorem~\ref{thm:2}}

The proof follows  closely that of Theorem~\ref{thm:1} with some obvious modifications. 
For illustration, we shall only argue part $(ii)$.  Let us first note that,  in parallel with \eqref{eq:an}, $b_n =(n/Var V_n^{(\a)})^{1/2}\to \infty.$ Indeed, since
\[ b_n^2\ge \frac{n}{E (V_n^{(\a)})^2}\ge \frac{n}{2\sum_{ij} p_{ij}\left[\left(\frac{q_i}{p_i}\right)^{2(1-\a)}+\left(\frac{p_j}{q_j}\right)^{2\a}\right]}\ge \frac{n}{2\sum \left(\frac{q_i}{p_i}\right)^{1-\a} +2\sum \left(\frac{p_i}{q_i}\right)^{\a} },\] therefore
\begin{equation}\label{eq:bn} 2b_n\ge  \frac{(nVar V_n^{(\a)} )^{1/2}} 
{\sum \left(\frac{q_i}{p_i}\right)^{1-\a} +\sum \left(\frac{p_i}{q_i}\right)^{\a} } \to \infty\end{equation} due to \eqref{eq:pc2}.
Next, we show that the limiting distribution  is determined by the projection $V_n^{(\a)}$. By the bivariate Taylor expansion, one obtains
\[
b_n(\S(\bhp,\bhq)-\S(\bp,\bq)) = b_n \left( \a \sum_i (q_i/p_i)^{1-\a} (\p_i-p_i) + (1-\a)\sum_i (p_i/q_i)^{\a} (\q_i-q_i)\right )+ b_nR_n\] where 
\[ R_n=\left.\sum_i \sum_{\{(k,l): k,l\ge 0, k+l=2\}} \frac{\partial^2 \tilde{p}_i^{\a}\tilde{q}_i^{1-\a}}{\partial^k \tilde{p}_i\partial^l \tilde{q}_i}\right|_{(\tilde{p}_i,\tilde{q}_i)=(p_i,q_i)+\th_i(\p_i-p_i,\q_i-q_i)} \frac{(\p_i-p_i)^k(\q_i-q_i)^l}{k!\, l!}
\] and   $ |\th_i|\le 1$ for all $i$. Since  for the mixed derivatives term,  by virtue of  the  elementary inequality $2ab\le a^2+b^2$ (with 
$a=\alpha(\hat{p}_i-p_i)/\tilde{p}_i$  and $b=(1-\alpha)(\hat{q}_i-q_i)/\tilde{q}_i$) we have $$2\a(1-\a)\sum \tilde{p}_i^{\a-1}\tilde{q}_i^{-\a}(\p_i-p_i)(\q_i-q_i)\le \a^2\sum \tilde{p}_i^{\a-2}\tilde{q}_i^{1-\a}(\p_i-p_i)^2+(1-\a)^2\sum \tilde{p}_i^{\a}\tilde{q}_i^{-\a-1}(\q_i-q_i)^2, $$ 
therefore
$$b_nR_n\le  2b_n (R_n^{(1)}+R_n^{(2)}) $$ with 
\[R_n^{(1)}=\sum_i p_i^{\a} q_i^{1-\a} \left(\th_i\left(\frac{\p_i}{p_i}-1\right)+1\right)^{\a-2} \left(\th_i\left(\frac{\q_i}{q_i}-1\right)+1\right)^{1-\a} \left(\frac{\p_i-p_i}{p_i}\right)^2 \]
and 
\[R_n^{(2)}=\sum_ip_i^{\a} q_i^{1-\a} \left(\th_i\left(\frac{\p_i}{p_i}-1\right)+1\right)^{\a} \left(\th_i\left(\frac{\q_i}{q_i}-1\right)+1\right)^{-\a-1} \left(\frac{\q_i-q_i}{q_i}\right)^2 .\]  Clearly, it suffices now to show only that 
$b_n R_n^{(i)}=o_p(1)$ for $i=1,2$. 
We only prove the second relation,  the other  one follows similarly.
Analogously as in   the proof of Theorem~\ref{thm:1} taking some small $\delta>0$ we have 
\begin{align*}P(b_nR_n^{(2)}>\ve) &\le P\left(b_nR_n^{(2)}>\ve,\max_i \left\{\left| \th_i\left(\frac{\p_i}{p_i}-1\right)\right| +\left| \th_i\left(\frac{\q_i}{q_i}-1\right)\right|\right\}\le\delta\right)\\
&+P\left(\max_i\left\{\left| \th_i\left(\frac{\p_i}{p_i}-1\right)\right|+\left|\th_i \left(\frac{\q_i}{q_i}-1\right)\right|\right\}> \delta\right)\\ &=:(I)+(II).\end{align*}
Apropos $(I)$,   for some universal ($n$-free and $\delta$-free) constant $C$ we have 
$$ (I) \le b_n\ve^{-1}  \sum_ip_i^{\a} q_i^{1-\a} \left(1+\delta\right)^{\a} \left(1-\delta \right)^{-\a-1} E\left(\frac{\q_i-q_i}{q_i}\right)^2 
\leq C\varepsilon^{-1} \sum p_i^{\a}q_i^{-\a}/(n Var  V_n^{(\a)})^{1/2}\to 0 $$
 by \eqref{eq:pc2}.  Apropos $(II)$, we have
$$ (II)\le P\left( \max_i \left| \left(\frac{\p_i}{p_i}-1\right) \right| >\delta/2\right)+P\left(\max_i \left|\left(\frac{\q_i}{q_i}-1\right)\right|>\delta/2\right)=:(IIa)+(IIb)$$
Note that for $d$ large enough so that $d\tau> 1$, in view of \eqref{eq:ld} and the Boole inequality bound combined with  the result of Lemma~\ref{lem:a1},
\begin{align*} (IIa) &\leq
 P\left( \max_i \left|\max\left\{ \left(\frac{p_i}{q_i}\right)^{\frac{\a}{2d}},\left(\frac{q_i}{p_i}\right)^{\frac{1-\a}{2d}}\right\} \left(\frac{\p_i}{p_i}-1\right)\right|>\delta/2\right)\leq  \left(\frac{2}{\delta}\right)^{2d}\sum \max\left\{\left(\frac{p_i}{q_i}\right)^\a, \left(\frac{q_i}{p_i}\right)^{1-\a} \right\}(np_i)^{-d}\\
 & \le \left(\frac{2}{\delta}\right)^{2d}n^{-1} \sum \left[\left(\frac{p_i}{q_i}\right)^\a +\left(\frac{q_i}{p_i}\right)^{1-\a}\right]\le  \left(\frac{2}{\delta}\right)^{2d}b_nn^{-1}\sum \left[\left(\frac{p_i}{q_i}\right)^\a +\left(\frac{q_i}{p_i}\right)^{1-\a}\right]\rightarrow 0\\
& \text{by \eqref{eq:bn} and \eqref{eq:pc2}. Similarly, } \\
 (IIb) &\le P\left( \max_i \left|\max\left\{\left(\frac{p_i}{q_i}\right)^{\frac{\a}{2d}},\left(\frac{q_i}{p_i}\right)^{\frac{1-\a}{2d}}\right\} \left(\frac{\q_i}{q_i}-1\right)\right|>\delta/2\right)\le  \left(\frac{2}{\delta}\right)^{2d}b_nn^{-1}\sum \left[\left(\frac{p_i}{q_i}\right)^\a +\left(\frac{q_i}{p_i}\right)^{1-\a}\right]\rightarrow 0\end{align*} and therefore $b_nR_n^{(2)}=o_p(1)$ and by a similar argument  $b_nR_n^{(1)}=o_p(1)$. Consequently, $b_nR_n=o_p(1).$
 Finally, since the distribution of
$\a \sum_i (q_i/p_i)^{1-\a} (\p_i-p_i) + (1-\a)\sum_i (p_i/q_i)^\a (\q_i-q_i)$ is equal to  that of  $\sum_{k=1}^n (V_{k}^{(n)}-EV_{k}^{(n)})/n$ where $V_{k}^{(n)}$ are independent and distributed according to $V_n^{(\a)}$ given in \eqref{eq:v}, the result follows by  Lemma~\ref{lem:lc}.
\section{Proofs for   Degenerate Projections (Section~3.2)}

\subsection*{Auxiliary Results} The following lemma is cited after  \citet[][Theorem~ 4.7.3, page 162]{BK}.
 \begin{lemma}[Degenerate $U$-statistic CLT]\label{lem:bk} Let $X_1\ldots,X_n$ be a sequence of iid random elements and let $$U_n(X_1,\ldots,X_n)=\binom{n}{2}^{-1}\sum_{1\le k<l\le n} h_n(X_k,X_l)$$ be a $U$-statistic of order two with a symmetric, real-valued kernel  $h_n(x,y)$  which depends on $n$ and satisfies $Eh_n(X,y)=0$. Denote also 
$\Psi_n(y,z)=E(h_n(X,y)h_n(X,z)).$
Assume that  $Eh_n^4<\infty$ and set $\sigma_n^2=Eh_n^2$. If the conditions  
\begin{align}\label{eq:1a}
&n^{-1}\sigma_n^{-4}Eh_n^4\to 0\\
& \sigma_n^{-4} E \Psi_n^2\to 0
\label{eq:2a}
\end{align}
are satisfied, then 
$$ n U_n/(\sqrt{2}\sigma_n)\Rightarrow N(0,1).$$ \qed
\end{lemma}

\subsection*{Proof of Lemma~\ref{lem:ch1}}
We start by  showing $(i)$. To this end, let $X_k$  for $k=1,\ldots,n$ be iid   single trial  multinomial variables  with parameter $\bp$ and denote 
$I(X_k=i)=\delta_i(X_k)$. 
Note  the identity
\begin{align*} \chp(\bhp,\bp)-m+1 &=n\sum p_i^{-1}(n^{-1}\sum_k \delta_i(X_k)-p_i)^2- m+1 \\ &= n^{-1}\sum p_i^{-1}\sum_{k\ne l} \delta_i(X_k)\delta_i(X_l) +n^{-1}\sum p_i^{-1}\sum_k \delta_i(X_k)-n-m+1\\
&=n^{-1}\sum_{k\ne l} (\sum p_i^{-1} \delta_i(X_k)\delta_i(X_l) -1)+n^{-1}\sum_k (\sum p_i^{-1} \delta_i(X_k)-m)\\
&=(n-1)U_n^{(1)}+R_n^{(1)}.
\end{align*}
Here $U_n^{(1)}$ is the  $U$-statistic with  order-two kernel $h_n(X_1,X_2)=\sum_i p_i^{-1} I(X_1=X_2=i)-1$ which is degenerate, i.e., satisfies $Eh_n(X_1,x_2)=0$. The remaining term  $R_n^{(1)}=n^{-1}\sum_k (V_k-EV_k)$,  where $V_k$ are iid, equidistributed with, say, $V$  such that  $P(V=p_i^{-1})=p_i$. We will argue that  \begin{equation}\label{eq:u1} (n-1)U_n^{(1)}/\sqrt{2(m-1)}\Rightarrow N(0,1)\end{equation} and  $R_n^{(1)}/\sqrt{m}=o_p(1).$
The second relation follows easily, since $m^{-1} Var R_n^{(1)}=(nm)^{-1} Var V=(nm)^{-1} (\sum p_i^{-1} -m^2)\to 0$ by assumption \eqref{eq:con2}. 

The convergence \eqref{eq:u1} will follow from Lemma~\ref{lem:bk} and Slutsky's theorem upon checking the conditions \eqref{eq:1a} and \eqref{eq:2a}. To this end   note that in the notation of Lemma~\ref{lem:bk} $\sigma_n^2=Var\, h_n=m-1$ since $Var [h_n(X_1,X_2)]=E[h_n(X_1,X_2)]^2=\sum p_i E(p_i^{-1} I(X_1=i)-1)^2=\sum p_iE(p_i^{-2}I(X_1=i)+1-2p_i^{-1}I(X_1=i))=m-1$. Similarly, 
$Eh_n^4=\sum p_i E(p_i^{-1}I(X_1=i)-1 )^4=\sum (p_i^{-2}-4p_i^{-1})+6m-3.$  Therefore 
$$ n^{-1}Eh_n^4/\sigma_n^4\le \frac{\sum (p_i^{-2}-4p_i^{-1})+6m-3}{n(m-1)^2}\le Cn^{-1}\sum p_i^{-2}m^{-2}\to 0 $$  due to \eqref{eq:con2} and  thus  \eqref{eq:1a} follows.  In order to  verify \eqref{eq:2a}, consider  first 
$\Psi_n(x,y)=E[h_n(X_1,x)h_n(X_1,y)]=p_x^{-1}I(x=y)-1.$ Since $E\Psi_n^2(X_1,X_2)=
\sum p_i E\Psi_n^2(X_1,i)=\sum p_i E[p_i^{-1} I(X_1=i)-1]^2=m-1$, then we have $E\Psi_n^2(X_1,X_2)/\sigma_n^4\to~0$ and \eqref{eq:2a} follows as well. Hence \eqref{eq:u1} follows and yields the assertion $(i)$. 

Now consider part $(ii)$. In parallel to part $(i)$,  define (cf. Section~\ref{sec:2})
 $Z_k=(X_k,Y_k)$ for    $k=1,\ldots,n$ as  a sequence of independent bivariate random variables distributed according to $Z=(X,Y)$. Additionally, for $i=1,\ldots,m$, as before let $\delta_i(X_k)=I(X_k=i)$, as well as  $\delta_i(Y_k)=I(Y_k=i)$. Set also $\Delta_i(Z_k)=\Delta_i(X_k,Y_k)=\delta_i(X_k)-\delta_i(Y_k)$. Note that for given $i$ the $\Delta_i(Z_k)$'s  for $k=1,\ldots,n$ are independent  variables  distributed  according to 
 $$\Delta_i(Z) =\begin{cases}
      0& \text{with prob.}\ 1-2(p_i-p_{ii}), \\
      1 & \text{with prob.}\ p_i-p_{ii} \\
       -1 & \text{with prob.}\ p_i-p_{ii}. \\
\end{cases} $$ In particular,  $E\Delta_i(Z)=0$ and $E\Delta^2_{i}(Z)=2(p_i-p_{ii})$.  Recall that $\mu_n=\sum (1-p_{ii}/p_i)$ and consider 
\begin{align*}
&\ch_{2\bp}(\bhp,\bhq)-\mu_n=n^{-1}\sum (2p_i)^{-1}(\sum_k \Delta_i(Z_k))^2-\mu_n \\ 
&=n^{-1}\sum_{k\ne l} \sum (2p_i)^{-1}\Delta_i(Z_k)\Delta_i(Z_l)+n^{-1}\sum_k \sum (2p_i)^{-1} (\Delta_i^2(Z_k)-2(p_i-p_{ii}))\\ &=(n-1) U_n^{(2)}+R_n^{(2)}.
\end{align*} which parallels the representation in part $(i)$.
Regarding $R_n^{(2)}$ note that it is, as before,  the zero mean sum of independent variables with variance 
\begin{align*}  Var\, R_n^{(2)} &=Var [n^{-1}\sum_k \sum_i (2p_i)^{-1}  \Delta_i^2(Z_k)]=n^{-1} Var[\sum_i (2p_i)^{-1}  (\delta_i(X)-\delta_i(Y))^2]\\
&\le  n^{-1} \sum_{i\ne j} \left((2p_i)^{-1}+(2p_j)^{-1}\right)^2 p_{ij}\le 2n^{-1} \sum_{i\ne j} \left((2p_i)^{-2}+(2p_j)^{-2}\right) p_{ij}\\ &\le n^{-1} \sum p_i^{-1}, \end{align*}
where the first inequality  above is obtained by applying  the second moment bound and  noticing  that the inner sum  consists of either two or zero summands, according to  $X\ne Y$ or $X=Y$.  The condition \eqref{eq:con2} implies that 
\begin{equation}\label{eq:3}
R_n^{(2)}/\sqrt{m}=o_p(1).
\end{equation} Regarding $U_n^{(2)}$, note that it is a $U$-statistic in bivariate variables $Z_k$ with second order degenerate kernel  given by $h_n(Z_1,Z_2)=\sum (2p_i)^{-1} \Delta_i(Z_1)\Delta_i(Z_2)$. Hence, in the notation of Lemma~5, $\sigma_n^2=E h_n^2(Z_1,Z_2)= E\left[\sum (2p_i)^{-1} \Delta_i(Z_1)\Delta_i(Z_2) \right]^2=\sum (2p_i)^{-2} (E\Delta_i^2(Z_1))^2+ \sum_{i\ne j} (2p_i)^{-1}(2p_j)^{-1}(E[\Delta_i(Z_1)\Delta_j(Z_1)])^2.$ Since we have $E\Delta_i^2(Z_1)=2(p_i-p_{ii})$ and  $E[\Delta_i(Z_1)\Delta_j(Z_1)]=-E[I(X_1=i,Y_1=j)+I(X_1=j,Y_1=i)]=-(p_{ij}+p_{ji})$, it follows that  $\sigma_n^2=\gamma_n^2$ given in \eqref{eq:sn}. Recall  (Remark~\ref{rem4})  that  under our assumptions $\sigma_n/\sqrt{m}\to 1$ and thus \eqref{eq:3} implies  $R_n^{(1)}/\sigma_n=o_p(1)$. 

Now, in order to complete the proof as in part $(i)$, we  only need to show that the conditions \eqref{eq:1a} and \eqref{eq:2a} are satisfied for $U_n^{(2)}$, since   then by Lemma~\ref{lem:bk}  the statement similar to  \eqref{eq:u1} holds for part $(ii)$, namely  \begin{equation}\label{eq:u2} (n-1)U_n^{(2)}/(\sqrt{2}\sigma_n)\Rightarrow N(0,1).\end{equation} 
To this end note first that by the definition of $\Delta$ variables $\Delta_i(Z)=\Delta_i^{2p+1}(Z)$ and $\Delta_i^2(Z)=\Delta_i^{2p}(Z)$ for any integer $p\ge 1$. Additionally, since $Z$ is bivariate,  for any   distinct  indices $(i,j,k,l)$    we have  $E[\Delta_i(Z)\Delta_j(Z)\Delta^2_k(Z)]=E[\Delta_i(Z)\Delta_j(Z)\Delta_k(Z)\Delta_l(Z)]=0$. Consequently,
\begin{align*} E h_n^4 &= E[\sum (2p_i)^{-1} \Delta_i(Z_1)\Delta_i(Z_2)]^4=2^{-4} \{\sum p_i^{-4} E[\Delta_i^4(Z_1)\Delta_i^4(Z_2)] \\
&\qquad +6\sum_{i\ne j} p_i^{-2}p_j^{-2} E[\Delta_i^2(Z_1)\Delta_i^2(Z_2)\Delta_j^2(Z_1)\Delta_j^2(Z_2)] \\
&\qquad +4\sum_{i\ne j} p_i^{-3}p_j^{-1} E[\Delta_i^3(Z_1)\Delta_i^3(Z_2)\Delta_j(Z_1)\Delta_j(Z_2) ]+0\}\\
&= 2^{-4} \{\sum p_i^{-4} 4(p_i-p_{ii})^2 +6\sum_{i\ne j} p_i^{-2}p_j^{-2} (p_{ij}+p_{ji})^2+4\sum_{i\ne j} p_i^{-3}p_j^{-1} (p_{ij}+p_{ji})^2\}\\
&= 2^{-4} \{\sum p_i^{-4} 4(p_i-p_{ii})^2 +\sum_{i\ne j} (6 p_i^{-2}p_j^{-2}+ 4 p_i^{-3}p_j^{-1})(p_{ij}+p_{ji})^2\}.
\end{align*}
In view of the condition  \eqref{eq:bd} and Remark~\ref{rem4} as well \eqref{eq:con2}
$$\frac{E h_n^4}{ n\sigma_n^4} \le \frac{3}{2\,n (m-2B)^2} \sum (p_i^{-2} +m B^2+B^2p_i^{-1})\le \frac{C}{n m^2} \sum p_i^{-2}\to 0$$ for some universal $C>0$ and hence \eqref{eq:1a} holds. 
In order to argue  \eqref{eq:2a}, set $z_1=(x_1,y_1)$  and $z_2=(x_2,y_2)$. Then \begin{align*}\Psi_n(z_1,z_2) &=E [h_n(Z,z_1) h_n(Z,z_2)]\\ &=E\{[(2p_{x_1})^{-1}(\delta_{x_1}(X)-\delta_{x_1}(Y))-(2p_{y_1})^{-1}(\delta_{y_1}(X)-\delta_{y_1}(Y))]\\& [(2p_{x_2})^{-1}(\delta_{x_2}(X)-\delta_{x_2}(Y))-(2p_{y_2})^{-1}(\delta_{y_2}(X)-\delta_{y_2}(Y))]\}.\end{align*} It follows that 
\begin{align*}\Psi_n(z_1,z_2) &=(4p_{x_1}p_{x_2})^{-1}(2p_{x_1}I(x_1=x_2)-p_{x_1,x_2}-p_{x_2,x_1})\\&\qquad+(4p_{y_1}p_{y_2})^{-1}(2p_{y_1}I(y_1=y_2)-p_{y_1,y_2}-p_{y_2,y_1})\\&\qquad-(4p_{x_1}p_{y_2})^{-1}(2p_{x_1}I(x_1=y_2)-p_{x_1,y_2}-p_{y_2,x_1})\\ &\qquad-(4p_{y_1}p_{x_2})^{-1}(2p_{y_1}I(y_1=x_2)-p_{y_1,x_2}-p_{x_2,y_1}) \\ &=:R(x_1,x_2)+R(y_1,y_2)-R(x_1,y_2)-R(y_1,x_2)
\end{align*} where the last equality is the definition.
Now consider 
\begin{align*}
E \Psi^2_n(Z_1,Z_2) &= ER^2(X_1,X_2)+ER^2(Y_1,Y_2)+ER^2(X_1,Y_2)+ER^2(Y_1,X_2) \\&\qquad +2E[R(X_1,X_2)R(Y_1,Y_2)]-2E[R(X_1,X_2)R(X_1,Y_2)] \\&\qquad-2E[R(X_1,X_2)R(Y_1,X_2)]
-2E[R(Y_1,Y_2)R(X_1,Y_2)] \\ &\qquad-2E[R(Y_1,Y_2)R(Y_1,X_2)]+2E[R(X_1,Y_2)R(Y_1,X_2)]\\&\le 4(ER^2(X_1,X_2)+ER^2(Y_1,Y_2) +ER^2(X_1,Y_2)+ER^2(Y_1,X_2)). 
\end{align*} where the last inequality follows by applying the inequality $2|ab|\le a^2+b^2$ to the integrants  in the  cross-product terms.  To show that the quadratic terms above are of order  $O(m)$ recall the assumption \eqref{eq:bd} and note that  we have 
\begin{align*}  ER^2(X_1,X_2) &=\sum_{x_1,x_2} \frac{p_{x_1}p_{x_2}}{16\,p_{x_1}^2p_{x_2}^2}(2p_{x_1}I(x_1=x_2)-p_{x_1,x_2}-p_{x_2,x_1})^2\\ &\le \sum_{x_1,x_2} \frac{p_{x_1}p_{x_2}}{4\,p_{x_1}^2p_{x_2}^2}\,p^2_{x_1}I(x_1=x_2)+\sum_{x_1,x_2} \frac{p_{x_1}p_{x_2}}{16\,p_{x_1}^2p_{x_2}^2}4B^2(p_{x_1}p_{x_2})^2\le \frac{m+B^2}{4} 
\end{align*} and via a similar argument it is easy to see that this  bound applies also  to the remaining  quadratic terms. Thus recalling Remark~\ref{rem4}
$$ E\Psi_n^2(Z_1,Z_2)/\sigma_n^4\le 4 (m+B^2)/(m-2B)^2\to 0.$$
 Therefore both  \eqref{eq:1a} and  \eqref{eq:2a} are satisfied and consequently \eqref{eq:u2} holds.    In view of   \eqref{eq:3}, 
  the proof of part $(ii)$ is completed. 
\qed
\subsection*{Proof of Theorem~\ref{u1}}

Consider first part  $(i)$. 
By  Taylor's expansion (note that the first term vanishes) 
\begin{equation}\label{eq:tay2} \S(\bhu)-\S(\bu)=\sum \u_i^\a-m^{1-\a}=0+m^{1-\a}n^{-1} \binom{\a}{2} \ch_\bu+ m^{1-\a}\binom{\a}{3} R_{n}
\end{equation}
where  $$R_{n}= \sum m(\u_i-m^{-1})^2 \left(m \u_i-1\right)\left(\theta_i\left(m\u_i-1\right)+1\right)^{\a-3}\quad \text{for some random }\theta_i\in (0,1). $$ 
Since    by Lemma~\ref{lem:ch1} and Remark~\ref{rem3} the properly normalized variable $ \ch_\bu$ is asymptotically normal under our assumptions, it only suffices to  show   that $\R_n=nR_n/\sqrt{m}=o_p(1)$.

Fixing   $\delta\in(0,1/2)$, for  any  $\ve>0$, we have \begin{align*}P(|\R_n|>\ve) &=P\left(|\R_n|>\ve,\max_i \left|\theta_i\left(m{\u_i}-1\right)\right|>\delta\right)+P\left(|\R_n|>\ve,\max_i \left|\theta_i\left(m{\u_i}-1\right)\right|\le\delta\right)\\ & \le P\left(\max_i \left|\theta_i\left(m\u_i-1\right)\right|>\delta\right)+P\left(|\R_n|>\ve,\max_i \left|\theta_i\left(m\u_i-1\right)\right|\le\delta\right)\\ &=:(I)+(II).\end{align*}
Note that   Lemma~\ref{lem:a1} and the Boole inequality  imply for $d$ large enough so as $d\tau>1$
$$ (I)\le P\left(\max_i \left|\left(m\u_i-1\right)\right|>\delta\right) \le C_{2d}\sum\,(m/\delta^2 n)^d \le C_{2d} \delta^{-2d} mn^{-\tau d}\to 0  $$  
by assumption.
Regarding $(II)$, note that $\delta<1/2$ and on the events $\{ \omega: \max_i \left|\left(m\u_i-1\right)\right|\le\delta \}$  we have the bound   
$|\R_n| \le  \delta \frac{Cn}{\sqrt{m}}\sum m(\u_i-m^{-1})^2,$ for some universal (i.e., $n$  and $\delta$ free) constant $C$,  so that 
\begin{align*} (II) &\le P\left( C\left|\frac{ n}{\sqrt{m}}\sum (m(\u_i-m^{-1})^2 -1/n)\right | >\ve/2\delta  \right) +P\left( C\max_i\left|m \u_i-1\right| \sqrt{m}>\ve/2 \right) \\ 
&= P\left(\frac{ Cn}{\sqrt{m}}\left| \ch_\bu -m/n\right|>\ve/2\delta \right) +m (2 C/\ve)^{2d} (m^2/n)^d. \end{align*} 
Consequently, from the above considerations and Lemma~\ref{lem:ch1} (denoting as before a  standard normal variable by $N$) it follows that 
$$\limsup_n P(|\R_n| >\ve)\le \limsup_n\ (I) +\limsup_n\ (II)= P(|N|>\ve(2C\delta)^{-1})\le \gamma$$
for any small $\gamma>0$ with  $\delta$ sufficiently small. Therefore  $\limsup_n P(|\R_n| >\ve)=0$, which  completes the proof of $(i)$. 
Consider now the assertion $(ii)$. The argument here is very similar to that of part $(iii)$ in Theorem~\ref{thm:1} and  we only sketch  it out, for the sake of brevity. Denote $c_n =1+\binom{\a}{2}\frac{m}{n}$.  It follows from $(i)$ that    \begin{equation*}\label{eq:wl}c_n^{-1}\S(\bhu)/\S(\bu)-1=\sum c_n^{-1}(m\u_i)^\a/m-1=o_p(1).\end{equation*} Hence, by virtue of  the Taylor expansion \eqref{eq:log} 
\begin{align*}
&\frac{n}{\a \sqrt{(m/2)}}\left[\H(\bhu) -\log m- (1-\a)^{-1}\log c_n \right]\nn\\
 &=\frac{\sqrt{2}n}{\a(1-\a)c_n\sqrt{m}}\left[ \frac{\sum (m\u_i)^\a}{m}-c_n\right]+T_n\end{align*} where $T_n$ stands now for the scaled quadratic term in the log expansion \eqref{eq:log}. Note that 
  the assertion $(ii)$ follows as soon as we show that $T_n=o_p(1)$. Similarly as in the proof of $(i)$ above,  it follows that on the events 
$\{|c_n^{-1}\S(\bhu)/\S(\bu)-1|\le \delta\}$ for $0<\delta<1/2$,  we have $|T_n| \le \frac{Cn}{\sqrt{m}}\left[ c_n^{-1}\sum (m\u_i)^\a/m-1\right]^2$ for sufficiently large $n$ and a universal (free of $n$ and $\delta$,  as above) constant $C$. Therefore for any $\ve>0$ 
\begin{align*}P(|T_n|>\ve) & =  P(|T_n|>\ve,|c_n^{-1}\S(\bhu)/\S(\bu)-1|\le \delta)+ P(|T_n|>\ve,|c_n^{-1}\S(\bhu)/\S(\bu)-1|>\delta) \\
&\le  P(|T_n|>\ve,|c_n^{-1}\S(\bhu)/\S(\bu)-1|\le \delta)+ P(|c_n^{-1}\S(\bhu)/\S(\bu)-1|>\delta)\\
&\le P\left(\frac{Cn}{\sqrt{m}}\left| c_n^{-1}\sum (m\u_i)^\a/m-1\right| >\ve/\delta\right)+o(1)=o(1)
\end{align*}
 using part $(i)$ and the fact that $c_n\to 1$ and $\delta>0$ may be arbitrarily small. The result follows. \qed
 
 \subsection*{Proof of Theorem~\ref{u2}} 
We shall only prove part $(i)$ since part $(ii)$ then follows  similarly as in  Theorems~\ref{thm:1} and \ref{thm:2}  and part $(ii)$ of  Theorem~\ref{u1}.
 Note 
 \[
\S(\bhp,\bhq)-1 =  \frac{\a(\a-1)}{n} \ch_{2\bp}(\bhp,\bhq) + R_n\] where 
\begin{align*} R_n &=\left. \sum_{\{(k,l): k,l\ge 0, k+l=3\}} \sum_i\ \frac{\partial^2 \tilde{p}_i^{\a}\tilde{q}_i^{1-\a}}{\partial^k \tilde{p}_i\partial^l \tilde{q}_i}\right|_{(\tilde{p}_i,\tilde{q}_i)=(p_i,q_i)+\th_i(\p_i-p_i,\q_i-q_i)} \frac{(\p_i-p_i)^k(\q_i-q_i)^l}{k!\, l!}\\
&=:  \sum_{\{(k,l): k,l\ge 0, k+l=3\}} \binom{\a}{k}\binom{1-\a}{l}R_n(k,l)
\end{align*} and  for all $i$ $ |\th_i|\le 1$. Due to \eqref{eq:bd} and Remark~\ref{rem4} as well as Lemma~\ref{lem:ch1} $(ii)$  it suffices to show that $\tilde{R}_n(k,l)=nR_n(k,l)/\sqrt{m}=o_p(1)$ for $k\ge 0$, $l\ge 0$ such that  $k+l=3$. Due to the invariance of $R_n$ under swapping  $\tilde{p}_i^\a$ and   $\tilde{q}_i^{1-\a}$, it suffices to show the above only for the pairs 
$(k=3,l=0)$ and $(k=2,l=1).$
To this end,  note 
\begin{align*} \tilde{R}_n(3,0) &=(n/\sqrt{m})\sum \tilde{p}_i^{\a-3}\tilde{q}_i^{1-\a}(\p_i-p_i)^3\\ &=\frac{n}{\sqrt{m}}\sum \left( 1+\theta_i \frac{\q_i-q_i}{q_i}\right)^{1-\a}\left( 1+\theta_i \frac{\p_i-p_i}{p_i}\right)^{\a-3} p_i^{-1}(\p_i-p_i)^2\left( \frac{\p_i}{p_i}-1\right)\end{align*} and $|\theta_i|\le 1$ for all $i$. For $0<\delta<1/2$ and $\ve>0$, let $A_n(\delta)= \{\omega: \max |\p_i/p_i-1|>\delta$ or $\max |\q_i/q_i-1|>\delta\}$
\begin{align*}
P(\vert\tilde{R}_n(3,0)\vert>\ve) &=P(\{\vert\tilde{R}_n(3,0)\vert>\ve\} \cap A_n(\delta))+P(\{\vert\tilde{R}_n(3,0)\vert>\ve\} \cap A_n^c(\delta))\\  &\le P(A_n(\delta))+ P(\{\vert\tilde{R}_n(3,0)\vert>\ve \}\cap A_n^c(\delta)).
\end{align*}Note that with large $n$ $P(A_n(\delta))\le \gamma $ for  arbitrarily small $\gamma>0$  due to 
$$P(A_n(\delta))\le P(\max |\p_i/p_i-1|>\delta) + P(\max |\q_i/q_i-1|>\delta)\le \gamma/2+\gamma/2$$ which follows by Lemma~\ref{lem:a1} and the application of Boole's bound, as before.  Recall from   Lemma~\ref{lem:ch1} that  $\mu_n=\sum (1-p_{ii}/p_i)$. Then  
\begin{align*} P(\{\vert\tilde{R}_n(3,0)\vert>\ve\} \cap A_n^c(\delta)) &\le P\left(\frac{n}{\sqrt{m}}\frac{(1+\delta)^{1-\a}}{(1-\delta)^{3-\a}}
 \max |\p_i/p_i-1| \sum (\p_i-p_i)^2/p_i  >\ve \right)\\ &\le  P\left(\frac{2n}{\sqrt{m}}
 \left|\sum (\p_i-p_i)^2/p_i -\mu_n/n \right| >\ve/2 \delta  \right) \\ &+P\left(\frac{2\mu_n}{\sqrt{m}} 
 \max |\p_i/p_i-1|  >\ve/2 \right)=:(Ia)+(Ib).\end{align*}    Note  that \eqref{eq:cond3} implies in particular \eqref{eq:con2} and therefore due to the CLT  result in $(i)$ of Lemma~\ref{lem:ch1}  we have $(Ia)\le \gamma/2$ for arbitrarily small $\gamma>0$ with $n$ large enough, whereas for $(Ib)$    
$$ (Ib)\le  P\left(\sqrt{m} 
 \max |\p_i/p_i-1|  >\ve \right)\le C m \left(\frac{m}{n \min p_i}\right)^d\le C m n^{-\tau d}\le \gamma/2 $$ for sufficiently large $d$, due to \eqref{eq:cond3},  the Boole inequality and   Lemma~\ref{lem:a1} (cf. previous proof). Consequently,  for arbitrarily small $\gamma$ and large $n$ 
$$ P(\vert\tilde{R}_n(3,0)\vert>\ve)\le 2\gamma.$$
For the  term  $\tilde{R}_n(2,1)$ note 
\begin{align*} \tilde{R}_n(2,1) &=(n/\sqrt{m})\sum (\tilde{p}_i^{\a-2}/\tilde{q}_i^{\a})(\p_i-p_i)^2(\q_i-q_i)\\ &=\frac{n}{\sqrt{m}}\sum \left( 1+\theta_i \frac{\q_i-q_i}{q_i}\right)^{-\a}\left( 1+\theta_i \frac{\p_i-p_i}{p_i}\right)^{\a-2} p_i^{-1}(\p_i-p_i)^2\left( \frac{\q_i}{q_i}-1\right).\end{align*}  The argument as above then applies also to bounding from above the probability $P(\vert\tilde{R}_n(2,1)\vert>\ve)$ with the obvious modification that 
\begin{align*} P(\{\vert\tilde{R}_n(2,1)\vert>\ve\} \cap A_n^c(\delta)) &\le P\left(\frac{n}{\sqrt{m}} (1-\delta)^{-2}
 \max |\q_i/q_i-1| \sum (\p_i-p_i)^2/p_i  >\ve \right)\\ &\le  P\left(\frac{2n}{\sqrt{m}}
 \left|\sum (\p_i-p_i)^2/p_i -\mu_n/n \right| >\ve/2 \delta  \right) \\ &+P\left(\frac{2\mu_n}{\sqrt{m}} 
 \max |\q_i/q_i-1|  >\ve/2 \right)=:(IIa)+(IIb)\end{align*}
for sufficiently small $\delta>0$.  One may then show that $(IIa)\le \gamma/2$ and $(IIb)\le\gamma/2$ and thus  $P(\vert\tilde{R}_n(2,1)\vert>\ve)  \le 2\gamma$ and  part $(ii)$ follows. \qed

\end{document}